\RequirePackage{amsmath}
\documentclass[smallcondensed,draft]{svjour3}
\usepackage{amssymb}
\usepackage{tikz}
\tikzstyle{arrow}=[-angle 45]
\usepackage{tikz-cd}
\usepackage{wrapfig}
\usepackage[final]{listings}
\usepackage{mathptmx}
\usepackage{stmaryrd,mathtools}

\usepackage{MnSymbol}
\usepackage{hyperref}
\usepackage{url}
\usepackage{prooftree1}
\usepackage{proof}
\usepackage{comment}
\usepackage{listings}
\usepackage{xspace}
\usepackage[show]{ed}

\usepackage{combinators}
\usepackage{macros}

\begin{document}

\title{Building on the Diamonds between Theories: Theory Presentation Combinators}

\author{Jacques Carette \and Russell O'Connor \and Yasmine Sharoda}
\institute{J. Carette  \at
    Department of Computing and Software, McMaster University,
    Hamilton, Ontario, Canada, \\
    \email{carette@mcmaster.ca}\ ORCiD: 0000-0001-8993-9804
\and R. O'Connor \at
    Blockstream.com, \\ \email{roconnor@theorem.ca}
\and Y. Sharoda\\
    Department of Computing and Software, McMaster University,
    Hamilton, Ontario, Canada, \\
    \email{sharodym@mcmaster.ca}
}

\date{Received: date / Accepted: date}

\maketitle

\begin{abstract}

To build a large library of mathematics, it seems more efficient to take
advantage of the inherent structure of mathematical theories.
Various theory presentation combinators have been proposed, and some have been
implemented, in both legacy and current systems.  Surprisingly, the
``standard library'' of most systems do not make pervasive use of these
combinators.

We present a set of combinators optimized for reuse, via the tiny theories
approach. Our combinators draw their power from the inherent structure already
present in the \emph{category of contexts} associated to a dependently typed
language. The current work builds on ideas originating in CLEAR and Specware
and their descendents (both direct and intellectual). Driven by some design
criteria for user-centric library design, our library-building experience
via the systematic use of combinators has
fed back into the semantics of these combinators, and later into an updated
syntax for them.

\end{abstract}

\setcounter{page}{1}

\lstdefinelanguage{mathscheme}
    {morekeywords={Theory,combine,extend,by,type,axiom,
        implies,not,forall,and,not,or,Inductive,case,of,instance,via,
        defined-in,view,as},
    keywordstyle=\sffamily}

\lstset{language=mathscheme}

%\begin{bottomstuff}
%\permission
%\copyright\ 2012 Journal of Formal Reasoning
%\end{bottomstuff}

\section{Introduction}\label{sec:intro}

The usefulness of a mechanized mathematics system relies on the availability a
large library of mathematical knowledge, built on top of sound foundations.
While sound foundations contain many interesting intellectual challenges,
building a large library seems a daunting task simply because of its
sheer volume.  However, as has been
documented~\cite{MathSchemeExper,CaretteKiselyov11,DuplicationMizar},
there is a tremendous amount of redundancy in existing libraries.  Thus
there is some hope that by designing a good meta-language, we can reduce
the effort needed to build a library of mathematics.

Our aim is to build tools that allow library developers to take
advantage of commonalities in mathematics so as to build
a large, rich library for end-users, whilst expending much less actual
development effort than in the past.  Our means are not in themselves new:
we define various combinators for combining theories, and a semantics
for these.  What is new is that we leverage the surrounding structure
already present in dependent type theories, to help us decide which
combinators to create --- as they are, in some sense, already present.
Furthermore, we further craft our combinators to maximize reuse, of
definitions and concepts.
Our combinators are not solely for creating new theories, as previous
work all emphasize, but also automatically create re-usable connections
between theories, which makes it possible to transport results between
theories, thereby increasing automation~\cite{LittleTheories}.
This also represents the continuation of our work on
\emph{High Level Theories}~\cite{CaretteFarmer08} and
\emph{Biform Theories}~\cite{carette2018biform}
through building a network of theories, leveraging what we learned
through previous experiments~\cite{MathSchemeExper}.
These combinators have now been implemented three times%
\cite{TPCProto,next-700-git,diagrams_mmt} for three
different systems.

\subsection{The Context}
\begin{wrapfigure}[15]{L}{0pt}
  \begin{tikzpicture}
  \coordinate (M1) at (3, 5);
  \coordinate (M2) at (3, 4);
  \coordinate (M3) at (3, 3);
  \coordinate (M4) at (3, 2);
  \coordinate (M5) at (3, 1);
  \coordinate (M6) at (3, 0);
  \node[fill=blue!20,rounded corners] at (M1) {Magma};
  \node[fill=blue!20,rounded corners] at (M2) {Semigroup};
  \node[fill=blue!20,rounded corners] at (M3) {Pointed Semigroup};
  \node[fill=blue!20,rounded corners] at (M4) {Monoid};
  \node[fill=blue!20,rounded corners] at (M5) {AdditiveMonoid};
%  \node[fill=blue!20,rounded corners] at (M6) {Abelian Group};
  \draw [->,thick] (3,4.7) -- (3,4.3);
  \draw [->,thick] (3,3.7) -- (3,3.3);
  \draw [->,thick] (3,2.7) -- (3,2.3);
  \draw [->,thick] (3,1.7) -- (3,1.3);
%  \draw [->,thick] (3,0.7) -- (3,0.3);
  \end{tikzpicture}
  \caption{Theories}\label{fig:thystruct1}
\end{wrapfigure}
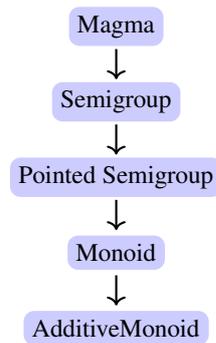

The problem we wish to solve is easy to state:  we want to shorten the
development time of large mathematical libraries.  But why would
mathematical libraries be any different than other software, where the
quest for time-saving techniques has been long but vain~\cite{MythicalManMonth}?
Because we have known since Whitehead's 1898 text
``A treatise on universal algebra''~\cite{whitehead1898treatise}
that significant parts of mathematics have a lot of structure, structure which
we can take advantage of. The flat list of $342$ structures gathered
by Jipsen~\cite{jipsen} is both impressively large, and could
easily be greatly extended.  Another beautiful source of structure in a
theory graph is that of \emph{modal logics}; Halleck's
web pages on Logic System Interrelationships~\cite{halleck} is quite eye
opening.

Figure~\ref{fig:thystruct1} shows what we are talking about:
The \emph{presentation} of the theory \thy{Semigroup} strictly contains that of
the theory \thy{Magma}, and so on\footnote{We are not concerned with
\emph{models}, whose inclusion go in the opposite direction.}.  It is
pointless to enter this information multiple times -- \emph{assuming} that it
is actually possible to take advantage of this structure.  Strict inclusions at
the level of presentations is only part of the structure: for example, we know
that a \thy{Ring} actually contains two isomorphic copies of \thy{Monoid},
where the isomorphism is given by a simple \emph{renaming}.  There are further
commonalities to take advantage of, which we will explain later in this
paper.  However, the natural structure is not linear like in
Figure~\ref{fig:thystruct1}, but full of diamonds, as shown in
Figure~\ref{fig:cube2_monoid}.  In Computer Science, this is known as
\emph{multiple inheritance}, and the diamonds in inheritance
graphs are much feared, giving rise to \emph{The Diamond Problem}%
\cite{jigsaw1992,traits2006,diamonds2011}, or fork-join inheritance
\cite{sakkinen1989disciplined}.  In our setting, we will find that
these diamonds are a blessing rather than a curse, because they statically
give sharing information, rather than being related to dynamic-dispatch.
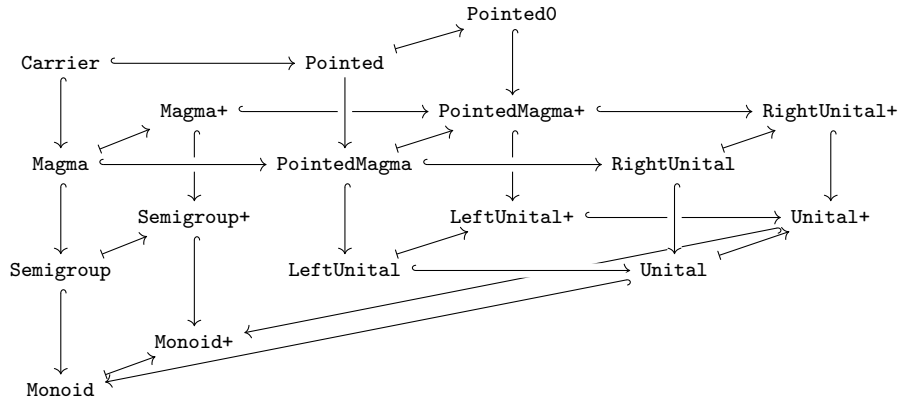
\begin{figure}
  \begin{tikzcd}[row sep=0.9em, column sep=0.3em]
    && & \verb|Pointed0| \arrow[dd,hook] & \\
    \verb|Carrier| \arrow[dd,hook] \arrow[rr,hook] & & \verb|Pointed| \arrow[ur,mapsto] & & \\
    & \verb|Magma+| \arrow[rr,hook] \arrow[dd,hook]& & \verb|PointedMagma+|
    \arrow[rr,hook]
    \arrow[dd,hook]& &
    \verb|RightUnital+|  \arrow[dd,hook]\\
    \verb|Magma| \arrow[ur,mapsto] \arrow[dd,hook]  & &
    \verb|PointedMagma| \arrow[dd,hook] \arrow[ur,mapsto] \arrow[from=uu, crossing over]
    \arrow[rr,hook,crossing over] \arrow[from=ll,hook, crossing over]
    & & \verb|RightUnital| \arrow[ur,mapsto] \\
    & \verb|Semigroup+| \arrow[ddd,hook] & &\verb|LeftUnital+| \arrow[rr,hook] &  &
    \verb|Unital+| \arrow[dddllll,hook] \\
    \verb|Semigroup|  \arrow[ddd,hook] \arrow[ur,mapsto]&  & \verb|LeftUnital|
    \arrow[ur,mapsto]
    \arrow[rr,hook,crossing over] &  & \verb|Unital| \arrow[dddllll,hook,crossing
    over]\arrow[ur,mapsto]
    \arrow[from=uu,hook,crossing over]& \\
    &&&& \\
    &  \verb|Monoid+| &  &&\\
    \verb|Monoid|  \arrow[ur,mapsto] && &&
  \end{tikzcd}
  \caption{Structure of the algebraic hierarchy up to Monoids}
  \label{fig:cube2_monoid}
\end{figure}

But is there
sufficient structure outside of the traditional realm of universal algebra, in
other words, beyond single-sorted equational theories, to make it worthwhile
to develop significant infrastructure to leverage that structure?  Luckily,
there is --- generalized algebraic theories~\cite{Cartmell} are a rich source,
which encompasses categories, bicategories, functors, etc as examples.

Our initial motivation, as with our predecessors, was to introduce
combinators to provide more expressive means to define theories. But
the structure of mathematics is expressed just as much through the theory
morphisms that arise between theories --- and not enough attention was paid,
in previous work, on these morphisms. These morphisms are always present in
the semantics of previous work, but not as much as first-class expressions
in the syntax of the system.  We feel that it is crucial that the induced
morphism be easily accessible, to enable transport of results and proofs
between theories.

We note that \textit{in practice}, when
mathematicians are \emph{using} theories rather than developing new ones, they
tend to work in a rather ``flat'' namespace~\cite{CaretteFarmer08}.
An analogy: someone working in
Group Theory will unconsciously assume the availability of all concepts from a
standard textbook, with their usual names and meanings.  As
their goal is to get some work done, whatever structure system builders have
decided to use to construct their system should not \textit{leak} into the
application domain.  They may not be aware of the existence of pointed
semigroups, nor should that awareness be forced upon them.  Thus we need features
that ``flatten'' a theory hierachy for some classes of end-users.
On the other hand,
some application domains do rely on the ``structure of theories'', so we cannot
unilaterally hide this structure from all users either.

\subsection{Contributions}
This paper is a substantial rewrite of~\cite{CaretteOConnorTPC}, where
a variant of the \emph{category of contexts} was used as our setting
for theory presentations.  There we presented a simple term language for
building theories, along with two (compatible) categorical semantics -- one in
terms of objects, another in terms of morphisms.  By using ``tiny theories'',
this allowed reuse and modularity.  We emphasized names, as the objects
we are dealing with are syntactic and ultimately meant for human consumption.
We also emphasized morphisms: while this is categorically obvious, nevertheless
the current literature is largely object-centric.  Put another way:
most of the emphasis in other work is on operational issues, or evolved from
operational thinking, while our approach is unabashedly denotational, whilst
still taking the names that appear inside theory seriously.

We extend the work in multiple ways.  We pay much closer attention to
the structure already present in the categorical semantics of dependent
type theories.  In particular, we extend our semantics to a fibration of
generalized extensions over contexts.  This is not straightforward: not
clobbering users' names prevents us from having a cloven fibration without a
renaming policy.  But once this machinery is in place, this allows us to
build presentations by lifting morphisms over embeddings, a very powerful
mechanism for defining new presentations.  There are obstacles to taking
the ``obvious'' categorical solutions: for example, having all pullbacks
would require that the underlying type theory have subset types, which is
something we do not want to force. This is why we insist on having users
provide an explicit renaming, so that they remain in control of the
names of their concepts.  Furthermore, equivalence of terms
needs to be checked when constructing mediating morphisms, which in some settings
may have implications for the decidability of typechecking. We also give 
complete algorithms --- that have now been implemented three times ---
as well as a type system.

While we are far from the first to provide such combinators, our
requirements are sufficiently different than previous work, that we arrive
at a novel solution, driven by what we believe to be an elegant semantics.

%\begin{itemize}
%\item Equivalence of terms ($\beta\eta\delta$-convertibility)
%    is important for checking commutativity when constructing mediating
%    morphisms. (potentially serious implications for decidablity of
%    typechecking.)
%\item Partial substition functors (aka life without a cloven fibration).
%\item Fibrations of entire diagrams of extensions over contexts.
%\end{itemize}

\subsection{Plan of paper}

We motivate our work with concrete examples in
section~\ref{sec:motivation}.  Section~\ref{sec:sem} lays out the basic
(operational) theory, with concrete algorithms.  The theoretical foundations of
our work, the fibered category of contexts, is presented in full detail in
section~\ref{sec:category}, along with the motivation for why we chose to
present our semantics categorically.  This allow us in section~\ref{sec:tpc2} to
formalize a language for theory presentation combinators and present a type
system for it.
We close with some discussion, related work and conclusions
in sections~\ref{sec:discussion}--\ref{sec:conclusion}.

\section{Motivation}\label{sec:motivation}

We outline our motivation for wanting theory presentation combinators, as
well as the design principles we use to guide our development.  Our
motivation seems to differ somewhat from previous attempts, which partly
explains why we end up with similar combinators but with different behaviour.
In this section, we use an informal syntax which should be understandable to
someone with a background in mathematics and type theory;
section~\ref{sec:tpc2} will formalize everything.
Unfortunately, the ``intuitive'' combinators work well on small examples,
but do not seem to work at scale.  We will highlight these problems
in this section, as a means to establish our requirements for a sound
solution.  This coherent semantics (developped in
sections~\ref{sec:sem}
and~\ref{sec:category}) will then lead us to rebuild our formal language,
including its syntax, in section~\ref{sec:tpc2}.

It is important to remember, throughout this section, that our
perspective is that of \emph{system builders}.  Our task is to form a
bridge (via software) between tasks that end-users of a mechanized
mathematics system may wish to perform, and the underlying (semantic) theory
concerned.  This bridge is necessarily syntactic, as syntax is the only
entity which can be symbolically manipulated by computers.  More importantly,
we must respect the syntactic choices of users, even when these choices are
not necessarily semantically relevant.  In other words, for theory presentations,
de Bruijn indices (for example) are unacceptable --- but so are generated names.

\subsection{Overview of Combinators}

Once mathematicians started using the axiomatic approach to Algebra, it became
clear that many theories are structurally related. Whitehead was the first to
formalize this as ``Universal Algebra''~\cite{whitehead1898treatise}. Ever since
then, it should have been well understood that defining a theory like
\tmop{Monoid}, at least when meant to be part of a reusable library of
mathematics, as a single ``module'' is not good practice. We should instead capture
the structural relations of mathematics in our library. There have been a
number of combinators designed previously for this task, and we distinguish
three of them: extension, rename and combination. We use \tmop{Monoid} and
related theories to illustrate them.

It might be useful to note that where we use the term ``combinator'',
others might have used ``construction''. And indeed, what we present
below are \emph{syntactic constructions} that take a \emph{presentation},
perhaps encoded as an algebraic data type, and produce another \emph{presentation}
in the same language. Each is meant to be an algorithmic construction.

\subsubsection{Extension}\label{subsec:extended-by}

The simplest situation is where the presentation of one theory is included,
verbatim, in another.  Concretely, consider \tmop{Monoid} and
\tmop{CommutativeMonoid}.
\[\label{Monoid}
\tmop{Monoid} \define
\begin{thyex}
\thyrow{U}{\tmop{Type}}
\thyrow{\_\cdummy\_}{U \rightarrow U\rightarrow U}
\thyrow{e}{U}
\thyrow{\tmop{right\ identity}}{\forall x:U. x\cdummy e = x}
\thyrow{\tmop{left\ identity}}{\forall x:U. e\cdummy x = x}
\thyrow{\tmop{associative}}
{\forall x,y,z:U. (x\cdummy y)\cdummy z = x\cdummy(y\cdummy z)}
\end{thyex}
\]
\[
\tmop{CommutativeMonoid} \define
\begin{thyex}
\thyrow{U}{\tmop{Type}}
\thyrow{\_\cdummy\_}{U \rightarrow U\rightarrow U}
\thyrow{e}{U}
\thyrow{\tmop{right\ identity}}{\forall x:U. x\cdummy e = x}
\thyrow{\tmop{left\ identity}}{\forall x:U. e\cdummy x = x}
\thyrow{\tmop{associative}}
{\forall x,y,z:U. (x\cdummy y)\cdummy z = x\cdummy(y\cdummy z)}
\thyrow{\tmop{commutative}}{\forall x,y:U. x\cdummy y = y\cdummy x}
\end{thyex}
\]

\noindent As expected, the only difference is that \tmop{CommutativeMonoid}
adds a \tmop{commutative} axiom.  Thus, given \tmop{Monoid}, it would be
much more economical to define
\[
\tmop{CommutativeMonoid} \define \extb{Monoid}{\tmop{commutative} : \forall x,y:U. x\cdummy y =
y\cdummy
x}
\]
The extension combinator is present in almost
every formal system.  It is sometimes viewed externally (as above) as
an extension, and sometimes internally, where \tmop{Monoid} is \textit{include}d
in \tmop{CommutativeMonoid}.

The construction is trivial, at least syntactically, as it is basically \emph{append}.
To ensure that the construction is semantically valid, one needs to check that the
newly added name is indeed new, and that the new type (here the commutativity
axiom) is well-typed in the context of the previous definitions.

\subsubsection{Renaming}
\label{subsec:renaming_motivation}
From an end-user perspective, our \tmop{CommutativeMonoid} has one flaw:
such monoids are frequently written \emph{additively} rather than
multiplicatively.  Let us call a commutative monoid written additively
an \emph{abelian monoid}, as we do with groups.  Thus it would be
convenient to be able to say
\[
\tmop{AbelianMonoid} \define\tmop{CommutativeMonoid}
 \left[\ \cdummy \rename +,\ e\rename 0 \right]
\]

But how are \tmop{AbelianMonoid} and
\tmop{CommutativeMonoid} related?  Traditionally, these are regarded as
\emph{equal}, for semantic reasons.  However, since we are dealing with
presentations, as syntax, we wish to regard them as \emph{isomorphic}
rather than equal\footnote{Univalent Foundations~\cite{hottbook} does not
change this, as we \textbf{can} distinguish the two, \emph{as presentations}.}
.  While working up to explicit isomorphism
is a minor inconvenience for the semantics, this enables us to respect
user choices in names.

Note that many existing systems do not have a renaming facility, and indeed
their libraries often contain both additive and multiplicative monoids that are
inequivalent.  This makes using the Little Theories method~\cite{LittleTheories}
very awkward.

\subsubsection{Combination}

Using these features, starting from \thy{Group} we might write
\[
\tmop{CommutativeGroup} \define \extb{Group}{ \tmop{commutative} : \forall x,y:U. x\cdummy y
= y\cdummy x }
\]
\noindent which is problematic: we lose the relationship that every commutative
group is a also commutative monoid.  In other words, we reduce our ability to
transport results ``for free'' to other theories, and must prove that these
results transport, even though the morphism involved is (essentially) the
identity.  We need a feature to express sharing.  Taking a cue from previous
work, we might want to say
\[
\tmop{CommutativeGroup} \define \combineover{CommutativeMonoid}{Group}{Monoid}
\]
\noindent  This can be read as saying that \thy{Group} and
\thy{CommutativeMonoid} are both ``extensions'' of \thy{Monoid}, and
\thy{CommutativeGroup} is formed by the union (amalgamated sum) of those
extensions.  In other words, by \textbf{over}, we mean to have a single
copy of \tmop{Monoid}, to which we add the extensions necessary for
obtaining \tmop{CommutativeMonoid} and \tmop{Group}.  This implicitly
assumes that our two \tmop{Monoid} extensions are meant to be orthogonal,
in some suitable sense.  
% We can build up to \tmop{CommutativeGroup}
% as shown in Figure~\ref{fig:comm_group}.
% \begin{figure}[htb]
%   \begin{tikzcd}[row sep=0.9em, column sep=1.5em]
%     \tmop{Carrier} \arrow[d,hook] & \\
%     \tmop{Magma} \arrow[d,hook] &  \\
%     \tmop{Semigroup} \arrow[d,hook] & \\
%     \tmop{PointedSemigroup} \arrow[d,hook] & \\
%     \tmop{Monoid} \arrow[d,hook] \arrow[r,hook] & \tmop{Group} \arrow[d,hook] \\
%     \tmop{CommutativeMonoid}  \arrow[r,hook] & \tmop{CommutativeGroup}\\
%   \end{tikzcd}
%   \caption{From \tmop{Carrier} to \tmop{CommutativeGroup}}
%   \label{fig:comm_group}
% \end{figure}

Unfortunately, while this ``works'' to build a sizeable library (say of the
order of $500$ concepts) in a fairly economical way, it is nevertheless
brittle.  Let us examine why this is the case.  It should be clear that by
\thy{combine}, we really mean \emph{pushout}. But a pushout is an operation
defined on $2$ morphisms (and implicitly $3$ objects); but our syntax gives the $3$
objects and leaves the morphisms implicit. Can we infer the morphisms
and prove that they are uniquely determined?  Unfortunately not:
these morphisms are (in general) impossible to infer, especially in the presence of
renaming.  As mentioned previously, there are two distinct morphisms from
\thy{Monoid} to \thy{Ring}, with neither being ``better'' or somehow
more canonical than the other.  In other words, even though our goal is to
produce \emph{theory presentations}, using pushouts as a
fundamental building block gives us no choice but to \emph{take morphisms
seriously}.

\subsubsection{Mixin}

There is one last annoying situation:
\begin{align*}
\tmop{LeftUnital} & \define \extb{PointedMagma}{ \tmop{leftIdentity} : \forall x:U. e\cdummy x
= x } \\
\tmop{RightUnital} & \define \extb{PointedMagma}{ \tmop{rightIdentity} : \forall x:U. x\cdummy e
= x } \\
\tmop{Unital} & \define \combineover{LeftUnital}{RightUnital}{PointedMagma}
\end{align*}
%There is one last annoying situation:
%\begin{lstlisting}[mathescape]
%$\tmop{LeftUnital} \define \extb{PointedMagma}{ \tmop{leftIdentity} : \forall x:U. e\cdummy x
%= x }$
%$\tmop{RightUnital} \define \extb{PointedMagma}{ \tmop{rightIdentity} : \forall x:U. x\cdummy e
%= x } $
% $\tmop{Unital} \define \combineover{LeftUnital}{RightUnital}{PointedMagma}$
%\end{lstlisting}

The type of \tmop{rightIdentity} arises from flipping the arguments to $\cdummy$ --- a
kind of symmetry.  Could we somehow capture this symmetry transformation and automate it?
We will not give suggestive syntax for this here, as a natural solution will emerge from
the semantics.  We will call this ``mixin'' as this construction bears a strong
resemblance to that of mixins in programming languages with traits.

\subsection{Morphisms}\label{sec:morphisms}

Our constructions do more than build new presentations, they also
describe how the symbols of the source theory can be mapped into
expressions of the target theory. For
extensions, this is an injective map. In other words,
\[
\tmop{CommutativeMonoid} \define \extb{Monoid}{\tmop{commutative} : \forall x,y:U. x\cdummy y =
y\cdummy x}
\]
creates $\tmop{CommutativeMonoid}$ and lets us see a
\tmop{Monoid} inside a \tmop{CommutativeMonoid}. What does it mean to
``see'' a \tmop{Monoid}?  It means to be able to have a definition of
all symbols of \tmop{Monoid}; in this case, this would be
\begin{align*}
\tmop{MtoCM} \define &
 \left[ \tmop{U} \mapsto \tmop{U},
        \cdummy \mapsto \cdummy,
        \tmop{e} \mapsto \tmop{e}, \right.\\
 & \left.{}\tmop{right\_identity} \mapsto \tmop{right\_identity},
           \tmop{left\_identity}  \mapsto \tmop{left\_identity}, \right.  \\
 & \left.{}\tmop{associative} \mapsto \tmop{associative} \right] :
  \tmop{CommutativeMonoid} \Rightarrow \tmop{Monoid}
\end{align*}
\noindent which is a tedious way of writing out the
identity morphism.

More economical is to instead write only what is to
be \emph{dropped}: $\tmop{MtoCM} \define \dispmap{\tmop{commutative}}$.

\paragraph{Renaming} For renaming, it is natural to require that the map on
names causes no collisions, as that would rename multiple concepts to be the
same.  While this is a potentially interesting operation on presentations, it
is not the operation that users have in mind for \emph{renaming}.  Thus
we will insist that renamings are \emph{bijections}.  These also
induce a map, in the opposite direction, that lets us see the source theory
inside the target theory.

\begin{wrapfigure}{r}{0pt}
	\begin{tikzpicture}[node distance=3.0cm, auto]
	\node (P) {$\left\{ V : \tmop{Type}\right\}$};
	\node (B) [right of=P] {$\left\{\ ? : \tmop{Type}\right\}$};
	\node (A) [below of=P] {$\left\{ U : \tmop{Type}\right\}$};
	\node (C) [below of=B] {$\left\{ W : \tmop{Type}\right\}$};
	\draw[->] (P) to node {$[V \rename\ ?]$} (B);
	\draw[->] (A) to node {$[U \rename V]$} (P);
	\draw[->] (A) to node [swap] {$[U \rename W]$} (C);
	\draw[->] (C) to node [swap] {$[W \rename\ ?]$} (B);
	\end{tikzpicture}
	\caption{The need for choosing names when combining theories.}
	\label{fig:renaming-problem}
\end{wrapfigure}
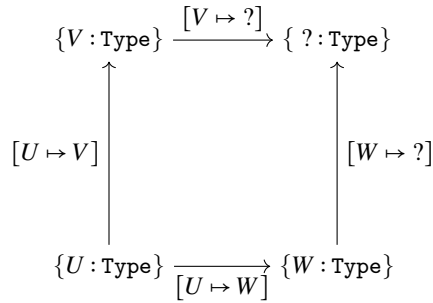

\paragraph{Combine}

Combinations do create morphisms as well, but unfortunately choosing names for
symbols in the resulting theory (when there are clashes) can be a
problem: there are simple situations where there is no canonical name for
some of the objects in the result.  For example, take the presentation of
\tmop{Carrier}, aka $\left\{ U\ :\ \tmop{Type}\right\}$ and the morphisms induced
by the renamings $U \rename V$ and $U \rename W$; while the result will
necessarily be isomorphic to \tmop{Carrier}, there is no canonical choice of
name for the end result.  This is one problem we must solve.
Figure~\ref{fig:renaming-problem} illustrates the issue.  
It also illustrates that we
really do compute \emph{amalgamated sums} and not simply syntactic union.
Because the names $U$, $V$ and $W$ are assumed to be meaningful to the user,
we do not wish to automatically compute a \emph{fresh} name for it,
to replace the $?$ in the Figure; while it is always feasible --- the
pushout does indeed exist --- it is not a good idea to invent names for
concepts that ought be be both meaningful and allow further compositional
extensions.

Extensions, renames and combines are sufficient to build a fairly sizable
library~\cite{MathSchemeExper}.  Extensions are used to introduce new
symbols and concepts.  Renames are used to ensure the ``usual'' name is
used in context; our library defines a \emph{binary operator} only once,
but there are dozens of different ones in use in mathematics. Renames take
care of renaming all uses in a theory, so that associativity, for example,
will apply to the `right' symbol.  Most theories in common use are an
amalgamation of concepts from smaller theories, and combine is the
simplest way to construct such theories.

In general, a map from one presentation $P$ to another $Q$ will be called a
\emph{presentation morphism}. Such a morphism can be written as a
sequence of assignments of valid terms of $Q$ for each symbol of $P$,
of the right type.
For example, one can witness that the additive naturals numbers form a monoid
(i.e. a morphism from $\tmop{Monoid}$ to $\tmop{Nat}$) by showing how
the terms of $\tmop{Nat}$ map to those of \tmop{Monoid}:
\begin{equation}\label{eq:natview}
\mathsf{view}\ \tmop{Nat}\ \mathsf{as}\ \tmop{Monoid}\ \mathsf{via}\
\left[ U \mapsto \N, \cdummy \mapsto +_{\N}, e \mapsto 0, ... \right]
\end{equation}
\noindent where we elide the names of the proofs.  The right hand side of
an assignment does not need to be a symbol, it can be any
well-typed term.  For example, we can have a morphism from
\tmop{Magma} to itself which maps the binary operation to its opposite:

\begin{equation}\label{eq:magmaview}
\begin{tikzpicture}[node distance=7.0cm, auto,baseline=(current bounding box.center)]
  \node (P) {$
\begin{thyex}
  \thyrow{U}{\tmop{Type}}
  \thyrow{\_\cdummy\_}{U \rightarrow U \rightarrow U}
\end{thyex} $};
  \node (B) [right of=P] {$
\begin{thyex}
  \thyrow{U}{\tmop{Type}}
  \thyrow{\_\cdummy\_}{U \rightarrow U \rightarrow U}
\end{thyex} $};
  \draw[->] (P) to node {$[U \mapsto\ U,
                          \cdummy \mapsto\ \mathsf{flip}\ \cdummy]$} (B);
\end{tikzpicture}
\end{equation}

Note that \emph{presentation morphisms} are \textbf{not} constructions, even
though our constructions give right to presentation morphisms. The relation
between them is like that of \emph{instances} in Haskell, where a
Haskell \texttt{class} is akin to a theory presentation, and an
\texttt{instance} is presentation morphism, from the given class to the
empty theory, i.e. ground terms in Haskell. A similar analogy holds for
signatures and structures in Ocaml and Standard ML, as well as traits and
instances in Scala.
\subsection{Polymorphism}
Recall the definition of \verb|AbelianMonoid|:
\[
\tmop{AbelianMonoid} \define\tmop{CommutativeMonoid}
\left[\ \cdummy \rename +,\ e\rename 0 \right]
\]
There is nothing specific to
\thy{CommutativeMonoid} in the renaming $\cdummy \rename +, e \rename 0$.
This is in fact a bijection of the four symbols involved --- and can be
applied to any theory where the pairs
$(\cdummy,+)$ and $(e,0)$ have compatible signatures (including the case where
they are not present).  Similarly, \thy{extended\_by} defines a
``construction'' which can be applied to any presentation whenever all the
symbols used in the extension are defined.  In other words, a reasonable
semantics should allow us to name and reuse these operations.
While it is tempting to think that these
operations will induce some functors on presentations, this is not quite
the case: name clashes prevent that.
Such a clash occurs when we try
to extend a theory with a new symbol (say $+$) when that symbol already
exists, and might mean something else entirely. Similar issues occur
with renamings.

\subsection{Little Theories}\label{sec:littletheories}

An important observation is that \emph{contexts} of a type theory
(or a logic) contain the same information as a \textit{theory presentation}.
Given a context, theorems about specific structures can be constructed by
transport along theory morphisms~\cite{LittleTheories}.
For example, in the context of the definition of~\tmop{Monoid} (\ref{Monoid}),
we can prove that the identity element, $e$, is unique:
\[
    \forall e' : U. \left( \left( \forall x.e' \cdummy x = x \right) \vee \left(
    \forall x.x \cdummy e' = x \right) \right) \rightarrow e' = e
\]

In order to apply this theorem in other contexts, we can provide a theory
morphism from one presentation to another. For example, consider
semirings:

\begin{multline*}
    \tmop{Semiring} \define \\
     \left\{ \begin{array}{rcl}
       U & : & \tmop{Type}\\
       \left( + \right) & : & (U,U)\rightarrow U\\
       \left( \times \right) & : & (U,U)\rightarrow U\\
       0 & : & U\\
       1 & : & U\\
       \tmop{+} \tmop{associative} & : & \forall x,y,z : U. \left( x +
       y \right) + z = x + \left( y + z \right)\\
       \tmop{+} \tmop{commutative} & : & \forall x,y : U.x + y = y +
       x\\
       \tmop{+} \tmop{left} \tmop{identity} & : & \forall x : U. 0 + x
       = x\\
       \tmop{+} \tmop{right} \tmop{identity} & : & \forall x : U.x + 0
       = x\\
       \tmop{\times} \tmop{associative} & : & \forall x,y,z : U.
       \left( x \times y \right) \times z = x \times \left( y \times z
       \right)\\
       \tmop{\times} \tmop{left} \tmop{identity} & : & \forall x : U.
       1 \times x = x\\
       \tmop{\times} \tmop{right} \tmop{identity} & : & \forall x :
       U.x \times 1 = x\\
       \tmop{left} \tmop{distributive} & : & \forall x,y,z : U.x \times \left(
       y + z \right) = x \times y + x \times z\\
       \tmop{right} \tmop{distributive} & : & \forall x,y,z : U. \left( y + z
       \right) \times x = y \times x + z \times x
     \end{array} \right\}
\end{multline*}

There are two natural morphisms from \tmop{Monoid} to \tmop{Semiring},
induced by the renamings
$\left[\ \cdummy \rename +,\ e\rename 0 \right]$ and
$\left[\ \cdummy \rename \times,\ e\rename 0 \right]$.
Both of these can be used to transport our example theorem to prove that $0$
and $1$ are the unique identities of their respective associated binary
operations. There are also many more morphisms; for example, we could send
$(\cdummy)$ to $\lambda x,y:U. y\cdummy x$.

Thus, in general, we cannot infer morphisms --- there are simply too many
non-canonical choices.  We also do not want to have to write out all morphisms
in explicit detail, as the example of $\dispmap{\tmop{commutative}}$ shows.

\subsection{Tiny Theories}

Previous experiments~\cite{MathSchemeExper} in library building with
combinators showed that for compositionality,
it was best to use \emph{tiny} theories, i.e. adding a single concept
in at a time.  This is useful
both for defining pure signatures (presentations with no axioms) as well
as when defining properties such as commutativity.
Typically one proceeds by first defining the smallest
typing context in which the property can be stated.  For commutativity,
this is \tmop{Magma} -- which also turns out to be a
signature.  We can then obtain the theories we are more interested in
by extending a signature with a combination of the properties over
that base.

Concretely, suppose we want to construct the
presentation of \tmop{CommutativeSemiring} by adding the commutativity property
to \tmop{Semiring} (see \S\ref{sec:littletheories}).
As commutativity is defined
as an extension to \tmop{Magma}, we need a morphism from
\tmop{Magma} to \tmop{Semiring}.  This morphism will tell us (exactly!) which binary
operation we want to make commutative.  Here we would
map $U$ to $U$ and $\left( \cdummy \right)$ to $
\left( \times \right)$.  We can then combine
that with the injection from \tmop{Magma} to \tmop{CommutativeMagma}
to produce a \tmop{CommutativeSemiring} presentation.

\[ \tmop{CommutativeSemiring} \define \left\{ \begin{array}{rcl}
U & : & \tmop{Type}\\
\left( + \right) & : & U \rightarrow U \rightarrow U\\
\left( \times \right) & : & U \rightarrow U\\
0 & : & U\\
1 & : & U\\
. & . & .\\ % visual hack
\tmop{\times} \tmop{commutative} & : & \forall x y : U.x \times
y = y \times x\\
. & . & . % another visual hack
\end{array} \right\} \]

Note that this construction also requires a further renaming that maps
\tmop{commutative} to \tmop{\times\,commutative} in
order to avoid a name collision (as addition was
already commutative in \tmop{Semiring}).

Figure~\ref{fig:cube2_monoid} is representative of the actual
kinds of graphs we get when using \emph{tiny theories} systematically.

Note that we did not need general morphisms here:
the ones we needed were combinations of morphisms induced by extensions and renamings.

Nevertheless, whether tiny theories are used or not, when we want
to work with a \tmop{Group}, we do not really care about the details of
how the library developers constructed it. In particular, we might never
have heard of \tmop{Monoid} or \tmop{Semigroup}, never mind \tmop{Magma}
--- and yet the user should still be able to use, and understand
\tmop{Group}.
This is not the case for other approaches, such as the hierarchy in Agda\cite{agda_stdlib},
Lean\cite{lean2019} or Coq\cite{geuvers2002,spitters2011}.

Furthermore, if library developers change their mind, this should not
cause any downstream problems.  When one of the authors added \tmop{Magma}
to the Algebra hierarchy of Agda, this caused incompatibilities between
versions of the standard library.  Similar problems happened in the
Haskell ecosystem
when the \tmop{Functor}-\tmop{Applicative}-\tmop{Monad} proposal~\cite{wiki:haskell_hierarch}
was adopted, which introduced \tmop{Applicative} in between \tmop{Functor} and
\tmop{Monad}.

\subsection{Models}

It is important to remember that models are contravariant: while the
construction of a presentation proceeds from (say)
\tmop{Monoid} to \tmop{CommutativeMonoid},
the model morphisms
are from $\sem{}{\tmop{CommutativeMonoid}}$ to
$\sem{}{\tmop{Monoid}}$.  Theorems are contravariant with respect to
\emph{model} morphisms, so that they travel from
\tmop{Monoid} to \tmop{CommutativeMonoid}.

In this way a presentation morphism from a presentation $T$ \emph{to} the empty
theory presentation provides a model by assigning to every symbol of $T$ a
closed term of the ambient type theory (or logic).
These models are thus internal, rather than necessarily being
$\mathsf{Set}$-models.  For example, if our underlying logic can express
the existence of a type of natural numbers, $\mathbb{N}$, then the
morphism given by~(\ref{eq:natview}) can be used to transport our theorem
to prove that $0$ is the unique identity element for $+_{\mathbb{N}}$.

\subsection{Induced Requirements}

We review the various issues that arose in the informal presentation of
combinators in this section, and assemble them into a set of
requirements for ``new'' combinators.

First, we need to have a setting in which extensions, renamings and
combinations make sense.  We will need to pay close attention
to names, both to allow user control of names and prevent accidental
collisions.  To be able to maintain human-readable names for all concepts, we
will put the burden on the \emph{library developers} to come up with a
reasonable naming scheme, rather than to necessarily push that issue onto end
users.  Symbol choice carries a lot of
\emph{intentional} and \emph{contextual} information which is commonly used in
mathematical practice.

We saw the need for a convenient language to define different kinds of
morphisms, as well as keep track of the properties of these morphisms.
Name permutations, extensions, combinations and general morphisms have
different properties --- tracking these can greatly simplify their use
and re-use.  We need a lightweight syntax where easily inferred information
can be omitted in common usage.

As we want to use both the little theories and tiny theories
methods, our language (and semantics) needs to allow, even promote, that
style.  We will see that, semantically, not all morphisms have the same
compositional properties.  We will thus want to single out, syntactically,
as large a subset of well-behaved morphisms as possible, even though we know
we can't be complete (some morphisms will be well-behaved but require a
proof to show that, rather than being well-behaved because of meta-theoretical
properties).

For example, we earlier tried to provide an explicit base theory over which
\tmop{combine} can work; this works whenever there is a unique injection from
the base theory into the extensions. While this is frequently the case, this
is not always so.  Thus we need to be able to be explicit about which
injection we mean.  A common workaround~\cite{MMT,implicit2019}
is to use long and/or qualified names that ``induce'' injections more often ---
but this still does not scale.  Much worse, this has
the effect of leaking the details of how a presentation was constructed into
the names of the symbols of the new presentation.  This essentially prevents
later refinements, as all these names would change.  As far as we can tell,
any automatic naming policy will suffer from this problem, which is why we
insist on having the library developers explicitly deal with name clashes.  We
can then check that this has been done consistently.  In practice few renamings
are needed, so if we have special syntax for ``no renaming needed'' that is
consistent with the syntax when a renaming is required, the resulting system
should not be too burdensome to use.

We can summarize our requirements as follows:
\begin{itemize}
\item Names and symbols associated to concepts should all be user-determined,
\item It should be possible to use the usual mathematical symbols for
concepts,
\item Concepts should be defined once in the minimal context necessary for
their definition, and transported to their use site,
\item Identical concepts with different names (because of change of
context) should be automatically recognized as ``the same'' concept,
\item The choices made by library developers of the means of construction
of theory presentations should be invisible to end-users,
\item Changing the means of construction of a theory should remain invisible,
\item Constructions should be re-usable (whenever possible),
\item Meta-properties of morphisms should not be forgotten by the system.
\end{itemize}

\section{Basic Semantics}\label{sec:sem}

We present needed definitions from dependent type theory and start
to develop the categorical semantics of theory presentation combinators.
The principal reason to use category theory is that previous work
on categorical semantics of dependent type theory has essentially
established that the structure we need for our combinators to ``work''
is already present there.

Recall the basic observation first made in Section~\ref{sec:littletheories},
that theory presentations and contexts of a dependent type theory
are the same: a list of symbol-type pairs, where the types of latter
symbols may depend on earlier symbols to be well-defined. In this
section, we will use ``presentation'' and ``context'' interchangeably.

Presentations depend on a background dependent type theory, but are
agnostic as to many of the internal details of that theory. We outline
what we require, which is standard for such type theories:
\begin{itemize}
  \item An infinite set of variable names \vars.

  \item A typing judgement for terms $s$ of type $\sigma$ in a context
  $\Gamma$ which we write $\Gamma \vdash s : \sigma$.

  \item A kinding judgement for types $\sigma$ of kind $\kappa$ in a context
  $\context{\Gamma}$ which we write\\
   $\context{\Gamma} \vdash \sigma : \kappa : \Box$.  We further assume that the set
   of valid kinds $\kappa : \Box$ is given and fixed.

  \item A definitional equality (a.k.a. convertibility) judgement of terms
  $s_1$ of type $\sigma_1$ and $s_2$ of type $\sigma_2$ in a context $\context{\Gamma}$,
  which we write $\context{\Gamma} \vdash s_1 : \sigma_1 \equiv s_2 : \sigma_2$. \ We
  will write $\context{\Gamma} \vdash s_1 \equiv s_2 : \sigma$ to denote $\context{\Gamma} \vdash
  s_1 : \sigma \equiv s_2 : \sigma$.

  \item A notion of substitution on terms. Given a list of variable
  assignments $\assignment{x_i}{s_i}{i < n}$
  and an expression $e$ we write $\substitutiondef{e}{x_i}{s_i}{i < n}$
  for the term $e$ after simultaneous substitution of variables $\left\{ x_i
  \right\}_{i < n}$ by the corresponding term in the assignment.
\end{itemize}
\noindent We use the meta-variable $v$ to denote an assignment, and its
application to a term $e$ by $\substitution{e}{v}{}$.

\subsection{Theory Presentations}\label{sec:pres}
A theory presentation is a well-typed list of declarations.
Figure~\ref{fig:ctx} gives the formation rules.
We use $\syms{\context{\Gamma}}$ to denote the set of variables names of a
well-formed context $\context{\Gamma}$:
$$ \syms{\context{\varnothing}} = \EmptyThy \qquad
   \syms{\context{\Gamma}\ ;\ x : \sigma} = \syms{\context{\Gamma}} \cup \left\{ x \right\}
$$
\begin{figure}[ht]
    \begin{proofrules}
      \[ \ \justifies \varnothing \ \wfctx \]
      \[ \context{\Gamma}\ \wfctx \qquad \sigma \notin \syms{\context{\Gamma}}
    \qquad \context{\Gamma} \vdash \kappa : \Box \justifies
    (\context{\Gamma}\ ;\ \sigma : \kappa)\ \wfctx \]
      \[ \context{\Gamma}\ \wfctx \qquad x\notin \syms{\context{\Gamma}}
                       \qquad \context{\Gamma} \vdash \sigma : \kappa : \Box \justifies
         (\context{\Gamma}\ ;\ x : \sigma)\ \wfctx \]
    \end{proofrules}
\caption{Formation rules for contexts}
\label{fig:ctx}
\end{figure}

\subsection{Morphisms (Views)}

As outlined in Section~\ref{sec:morphisms}, a morphism
from a theory presentation $\theory{\Gamma}$ to a theory presentation $\theory{\Delta}$
is an assignment of well-typed $\theory{\Delta}$-expression to each declaration of
$\theory{\Gamma}$. The assigments transport well-typed terms in the context
$\context{\Gamma}$ to well-typed terms in $\context{\Delta}$, by substitution.
Figure~\ref{fig:views} gives the formation rules.
\begin{figure}[ht]
\begin{proofrules}
  \[ \context{\Delta}\ \wfctx \justifies \view{}{\varnothing}{\Delta} \]
  \[ (\context{\Gamma}\ ;\ x : \sigma)\ \wfctx \qquad
     \view{v}{\Gamma}{\Delta} \qquad
     \context{\Delta} \vdash r : \substitution{\sigma}{v}{} \justifies
     \view{v,x \mapsto r}{(\context{\Gamma}\ ;\ x : \sigma)}{\context{\Delta}} \]
\end{proofrules}
  \caption{Formation rules for morphisms (substitutions).}
  \label{fig:views}
\end{figure}

There is a subtle but important distinction between assignments $\lbag v
\rbag$ and morphisms, $\left[ v \right] : \Gamma \rightarrow \Delta$:
morphisms are typed, and thus $\Gamma$ and $\Delta$ are integral to the
definition, while the same assignment may occur in different morphisms.

Since most morphisms allow us to ``view'' one presentation inside another,
we will frequently refer to the morphisms as \emph{views}.

There is no quotienting done on morphisms. We do rely on the
underlying type theory to furnish us with a notion of
equivalence.

\subsubsection{Inclusions, Renames and Embeddings}\label{sec:extdefn}

In Section~\ref{subsec:extended-by}, we define a construction to extend a
presentation with new fields, and in Section~\ref{subsec:renaming_motivation}, renamings
were defined. Section~\ref{sec:morphisms}, details how these correspond
to morphisms acting solely on names, being respectively an inclusion and
a bijection.

We thus define an \emph{embedding}
to be a special kind of morphism, which we denote
\[ \tilde{\pi} : \viewtype{\Gamma}{\Delta} \]
where we require that $\tilde{\pi}$ is the morphism induced by 
$\pi:\vars\to\vars$ where:
\begin{itemize}
\item $\pi : \vars \to \vars$ has \emph{finite support} (i.e.
outside of a finite subset of \vars, $\pi\left(v\right) = v$).
\item $\pi$ is a bijection,
\item $\pi^{-1}$ restricts to an injective function $\pi^{-1} : |\Delta| \to |\Gamma|$.
\end{itemize}
The $\tilde{\ }$ on $\tilde{\pi}$ is a reminder that while $\pi$ is function
on names, $\tilde{\pi}$ is a morphism with special properties.
Note that inclusions are embeddings, with empty support~\cite{pitts2013nominal}.
Embeddings thus include both inclusions and renamings.

As we have mentioned previously,
\tmop{Ring} is an extension of \tmop{Monoid} in two
different ways, and hence both embeddings cannot be inclusions.
Inclusions will not be special in our formalism, other than being a case
of an embedding.  We draw
attention to them here as many other systems make inclusions play a very
special role. As we will see later, it is instead
\emph{display maps} which should be held as holding a special role.

\subsubsection{Composition}

Given two morphism $\view{v}{\Gamma}{\Delta}$ and $\view{w}{\Delta}{\Phi}$,
we can compose them $\compose{v}{w}{\Gamma}{\Phi}$.
If $\viewdef{v}{\substitutiondef{}{a}{r_a}}{a \in \syms{\Gamma}}$
then the composite morphism is

\[ \composedef{v}{w}{\substitutiondef{}{a}{\substitution{r_a}{w}{}}{a \in \syms{\Gamma}}}  \]

That this gives a well-defined notion of composition, and that it is
associative is standard~\cite{Cartmell,JacobsCLTT,taylor1999practical}.

\subsubsection{Equivalence of Morphisms}

Two morphisms with the same domain and codomain,
$\viewname{u}, \view{v}{\Gamma}{\Delta}$
are \emph{equivalent} if\\
$\Delta \vdash r_a : (\substitution{\sigma_a}{u}{})
  \equiv s_a : (\substitution{\sigma_{a}}{v}{})$
  where
\begin{eqnarray*}
  \theory{\Gamma} \define \substitution{}{a : \sigma_{a}}{a \in \syms{\Gamma}} \\
  \viewdef{u}{\substitution{}{a := r_a}{a \in \syms{\Gamma}}} \\
  \viewdef{v}{\substitution{}{a := s_a}{a \in \syms{\Gamma}}}
\end{eqnarray*}

\subsubsection{The category of theory presentations}

The preceeding gives the necessary ingredients to define the
\emph{category of theory presentations} $\catthry$, with
theory presentations as objects and views as morphisms.
The identity inclusions are the
identity morphisms.

Note that in~\cite{CaretteOConnorTPC}, we worked with $\catctx = \catthry ^{op}$, which is
traditionally called the \emph{category of contexts}, and is more often
used in categorical
logic~\cite{Cartmell,JacobsCLTT,taylor1999practical,PittsAM:catl}.
In our setting, and as is common in the context of specifications
(see for example~\cite{PuttingTheoriesTogether,Smith93,CoFI:2004:CASL-RM}
amongst many others), we prefer to take our intuition from \emph{textual
inclusion} rather than \emph{models}.  Nevertheless, for
the \emph{semantics}, we too will use $\C$, as this
not only simplifies certain arguments, it also makes our work easier to
compare to that in categorical logic.

\subsection{Combinators}\label{sec:combinators}

Given theory presentations, embedding and views, we can
can now define presentation and view combinators.  In fact, all combinators in
this section will end up working in tandem on presentations and views.  They
allow us to create new presentations/views from old, in a
more convenient manner than building everything by hand.

The constructions (operational semantics) will be spelled out in full detail,
and are directly implementable.
In the next section, we will give them a categorical semantics;
we make a few inline remarks here to help the reader understand why we choose
a particular construction.

\subsubsection{Renaming}

Given a presentation $\theory{\Gamma}$ and
an injective renaming function $\renfuntype{\pi}{\Gamma}$
we can construct a new theory presentation $\Delta$
by renaming $\Gamma$'s symbols: we will denote this action of $\pi$ on $\Gamma$
by $\applyrenfun{\pi}{\Gamma}$.  We also construct an embedding from $\tilde{\pi}
: \Gamma \rightarrow \applyrenfun{\pi}{\Gamma}$ which
provides a translation from $\Gamma$ to the constructed presentation
$\applyrenfun{\pi}{\Gamma}$. 
For this construction as a whole, we use the notation
\[ \renameDef{\renSource}{\renFun} \]
% AbelianMonoid = CommutativeMonoid [\circ \mapsto + , e \mapsto 0]
\noindent The rename function used in Section~\ref{subsec:renaming_motivation}, produces the theory 
presentation
\tmop{AbelianMonoid} and the embedding 
\lstinline[mathescape]|$\tilde{\pi}$ : $\tmop{CommutativeMonoid}$ $\to$ $\tmop{AbelianMonoid}$|

\subsubsection{Extend}
\label{sec:extension}
Given a theory presentation $\extSource$, a
name $a$ that does not occur in $\extSource$ and a well formed type $\sigma$ of
some kind $\kappa$, (i.e. $\extSource \vdash \sigma : \kappa : \Box$) we can
construct a new theory
presentation
$\viewdef{\extTarget}{\extSource ; a : \sigma}$
and the embedding
$\tilde{\idF} : \Gamma \rightarrow \Delta$.
More generally, given a sequence of fresh names,
types and kinds, $\left\{ a_i \right\}_{i < n}$, $\left\{ \sigma_i \right\}_{i
< n}$, and $\left\{ \kappa_i \right\}_{i < n}$ we can define a sequence of
theory presentations $\extSource_0 \define \extSource$ and $\extSource_{i + 1} \define
\extSource_i ; a_i : \sigma_i$ so long as $\extSource_i \vdash \sigma_i : \kappa_i :
\Box$. Given such as sequence we construct a new theory presentation $\extTarget
\define \extSource_n$ with the embedding $\tilde{\idF} : \Gamma \rightarrow \Delta$. 

As $\extTarget$ is the concatenation of $\extSource$ with
$\left\{ a_i : \sigma_i : \kappa_i \right\}_{i < n}$,  we will use
$\extSource\rtimes\extTarget^{+}$ to denote the target of this view whenever the
components of $\extTarget^{+}$ are clear from context.
$\extTarget^{+}$ is \emph{rarely a valid presentation}, as it 
usually depends on $\extSource$.  This is why we use an asymmetric
symbol $\rtimes$.

Note that general embeddings
$\extArrow{\pi}{\extSource}{\extTarget}$ as defined in
\S\ref{sec:extdefn}
can be decomposed into a renaming composed with an $\rtimes$, in other
words
$\extArrow{\pi}{\Gamma}{\Delta} = \tilde{u} ; \tilde{\idF}$
where
  $\extArrow{u}{\Gamma}{\applyrenfun{\pi}{\Gamma}}$ and
$\extArrow{\idF}{\applyrenfun{\pi}{\Gamma}}{\applyrenfun{\pi}{\Gamma} \rtimes \extTarget^{+}}$.
We will also write these as $\substitution{\Gamma}{u}{\Delta}$ when we do
not wish to focus on the pieces of the embedding.

Embeddings which are inclusions are traditionally called \emph{display maps} in
$\C = \TP^{op}$, and our
$\tilde{\idF} : \arrowthrytype{\Gamma}{(\Gamma; a : \sigma)}$
in $\catthry$ is denoted
by $\displaymap{\hat{a}}{(\extSource ; a : \sigma)}{\extSource}$ in $\C$
~\cite{taylor1999practical}, and $\delta_{a}$ in~\cite{JacobsCLTT}.

For notational convenience, we encode the construction above as
an explicit function to a \emph{record}
containing two fields, \texttt{pres} (for presentation) and \texttt{embed}
(for embedding).

\[\extensionDef{\extSource}{\extDecls}\]
\noindent where $\extDecls = \left\{a_{i}:\sigma_{i}:\kappa_{i}\right\}_{i<n}$. 
Section~\ref{subsec:extended-by} gives an example of an extension with $\Gamma = $ \tmop{Monoid} and 
$\Delta^+$ being the declaration of the \tmop{commutative} axiom.

\subsubsection{Combine}
\label{sec:combine}

To combine two embeddings
$\view{u_{\Delta}}{\Gamma}{\Delta}$
and
$\view{u_{\Phi}}{\Gamma}{\Phi}$ to give a presentation $\combineResult$, we
need to make sure the results agree on $\Gamma$
(to avoid cases like the one in Figure~\ref{fig:renaming-problem}). 
To insure this, we ask that two injective renaming functions
$\renfuntype{\pi_\Delta}{\Delta}$
and
$\renfuntype{\pi_\Phi}{\Phi}$
satisfying
\begin{equation}\label{eqn:renaming}
\pi_{\Delta} \left( x \right) = \pi_{\Phi} \left( y \right)
\Leftrightarrow \exists z \in \left| \Gamma \right|.\ x =
\substitution{z}{u_{\Delta}}{} \wedge y = \substitution{z}{u_{\Phi}}{}
\end{equation}
\noindent are also provided. $\pi_{\Delta}$ and $\pi_{\Phi}$ will be used to give a unique
name to the components of $\Gamma$ --- for example the carrier $U$ in
Figure~\ref{fig:renaming-problem}.

Suppose that the two embeddings decompose as
$\Delta = \substitution{\Gamma}{u_{\Delta}}{} \rtimes \Delta^+$
and $\Phi = \substitution{\Gamma}{u_{\Phi}}{}\rtimes\Phi^+$. 
Denote by $u_{\Delta}$ the action on $|\Gamma|$ of $\view{u_{\Delta}}{\Gamma}{\Delta}$,
and by $u_{\Phi}$ the action on $|\Gamma|$ of $\view{u_{\Phi}}{\Gamma}{\Phi}$.
Define

\[\combineResult \define \combineResult_0 \rtimes \left( \combineResult_{\Delta} \cup
\combineResult_{\Phi} \right)\] where
\begin{align*}
\combineResult_0
   & \; \define \applyrenfun{\left(u_{\Delta} ; \pi_{\Delta}\right)}{\Gamma} \\
%\substitution{\Gamma}{z \mapsto \pi_\Delta(\substitution{z}{u_\Delta}{})}{z
%  \in \syms{\Gamma}} \\
\combineResult_{\Delta}
%    & \define \substitution{\Delta^+}{x \mapsto \pi_\Delta(x)}{x \in \syms{\Delta}} \\
     & \define \applyrenfun{\pi_{\Delta}}{\Delta^{+}} \\
\combineResult_{\Phi}
%    & \define \substitution{\Phi^+}{y \mapsto \pi_\Phi(y)}{y \in \syms{\Phi}}
     & \define \applyrenfun{\pi_{\Phi}}{\Phi^{+}} \\
\end{align*}
Condition~\ref{eqn:renaming} is defined to insure that
$\combineResult_0 \equiv
   \applyrenfun{\left(u_{\Phi} ; \pi_{\Phi}\right)}{\Gamma}$
is also true. Similarly, by construction,
$\combineResult_0 \rtimes \left( \combineResult_{\Delta} \rtimes
\combineResult_{\Phi}\right)$ is equivalent to
$\combineResult_0 \rtimes \left( \combineResult_{\Phi} \rtimes \combineResult_{\Delta}\right)$; we denote
this equivalence class\footnote{In practice, theory presentations are rendered
(printed, serialized) using a topological sort where ties are broken alphabetically, so as to be
construction-order indedepent.}
of views by
$\combineResult_0 \rtimes \left( \combineResult_{\Delta} \cup \combineResult_{\Phi}\right)$.

The combination operation also provides embeddings
$\view{v_{\Delta}}{\Delta}{\combineResult}$ and $\view{v_{\Phi}}{\Phi}{\combineResult}$
where
$\viewname{v_{\Delta}}\define\tilde{\pi}_{\Delta}$ and
$\viewname{v_{\Phi}}\define\tilde{\pi}_{\Phi}$.
A calculation shows that $\composeviews{u_\Delta}{v_\Delta}$
is equal to $\composeviews{u_\Phi}{v_\Phi}$ (and not just
equivalent);  we denote this joint morphism
$\view{uv}{\Gamma}{\combineResult}$.
Furthermore,
combine provides a set of mediating views from the constructed theory
presentation $\combineResult$. Suppose we are given views
$\view{w_{\Delta}}{\Delta}{\Omega}$ and
$\view{w_{\Phi}}{\Phi}{\Omega}$
such that the the composed views
$\compose{u_\Delta}{w_\Delta}{\Gamma}{\Omega}$ and
$\compose{u_{\Phi}}{w_{\Phi}}{\Gamma}{\Omega}$ are equivalent. \ We can
combine $\viewname{w_{\Delta}}$ and $\viewname{w_{\Phi}}$ into a
mediating view
$\view{w_{\combineResult}}{\combineResult}{\Omega}$ where
\[
\viewdef{\viewname{w_\combineResult}}
{\substitution{}{\pi_{\Delta}(x) := \substitution{x}{w_\Delta}{}}{x \in \syms{\Delta}}
\cup
 \substitution{}{\pi_{\Phi}(y) := \substitution{y}{w_\Phi}{}}{y \in \syms{\Phi}}
}
\text{.}
\]
This union is well defined since if $\pi_{\Delta} \left( x \right) =
\pi_{\Phi} \left( y \right)$ then there exists $z$ such that
$x = \substitution{z}{u_\Delta}{}$ and $y = \substitution{z}{u_\Phi}{}$, in which case
$\substitution{x}{w_\Delta}{} = \substitution{\substitution{z}{u_\Delta}{}}{w_\Delta}{}$
and
$\substitution{y}{w_\Phi}{} = \substitution{\substitution{z}{u_\Phi}{}}{w_\Phi}{}$
are equivalent since by assumption $\composeviews{u_{\Delta}}{w_{\Delta}}$ and
$\composeviews{u_{\Phi}}{w_{\Phi}}$ are equivalent.
It is also worthwhile noticing that this construction
is symmetric in $\Delta$ and $\Phi$.

For this construction, we use the following notation, where we use the
symbols as defined above (omitting type information for notational
clarity)
\[
\comfun\left( u_{\Delta}, u_{\Phi}, \pi_{\Delta}, \pi_{\Phi}\right) \define
\left\{\begin{aligned}
  \mathtt{pres} & = \combineResult_0\rtimes\left(\combineResult_{\Delta} \cup \combineResult_{\Phi}\right) \\
  \mathtt{embed}_{\Delta} & = \left[v_{\Delta}\right] : \Delta\rightarrow\combineResult\\
  \mathtt{embed}_{\Phi} & = \left[v_{\Phi}\right] : \Phi\rightarrow\combineResult\\
  \mathtt{diag} & = \left[uv\right]:\Gamma\rightarrow\combineResult\\
  \mathtt{mediate} & = \lambda\ w_{\Delta}\ w_{\Phi}\ .\  w_{\combineResult}
\end{aligned}\right\}\]
The attentive reader will have noticed that we have painstakingly constructed
an explicit \emph{pushout} in $\catthry$.  There are two reasons to do this:
first, we need to be this explicit if we wish to be able to implement such
an operation.  And second, we do not want an arbitrary pushout, because we
do not wish to work up to isomorphism as that would ``mess up'' the names.
This is why we need user-provided injective renamings $\pi_{\Delta}$ and
$\pi_{\Phi}$ to deal with potential name clashes.  If we worked up to
isomorphism, these renamings would not be needed, as they can always be
manufactured by the system -- but then these are no longer necessarily related
to the users' names. Alternatively, if we use \emph{long names} based on the
(names of the) views, the method used to construct the presentations and views
``leaks'' into the names of the results, which we also consider undesirable.

\paragraph{Example} Consider combining the two embeddings\\
$[u_\Delta] : \tmop{Magma} \to \tmop{Semigroup} = \tilde{\idF}$ and 
$[u_\Phi] : \tmop{Magma} \to \tmop{AddMagma} = [\tmop{U} \mapsto \tmop{U} , \circ \mapsto +]$; 
in the above, this makes $\Gamma = \tmop{Magma}$, $\Delta = \tmop{Semigroup}$ and
$\Phi = \tmop{AddMagma}$.
Choosing $\pi_\Delta = \pi_\phi = \idF$ does not satisfy
condition~\ref{eqn:renaming}. The problem is that the two embeddings
$[u_\Delta]$ and $[u_\Phi]$ disagree on the name of the binary operation.
Thus the user must provide a renaming; for example, the user might
choose $+$ and define \tmop{AddSemigroup}, by using $\pi_\Delta = [\circ
\mapsto +,  \tmop{associativity}\_\circ \mapsto
\tmop{associativity}\_+]$; $\pi_{\phi}$ can remain $\idF$. (We would prefer for the user to
only need to specify $[\circ \to +]$, and for the system to infer
$[\tmop{associativity}\_\circ \mapsto \tmop{associativity}\_+]$ but we leave this
for future work). The algorithm then proceeds to compute the expected
\tmop{AddSemigroup} and accompanying embeddings.

\subsubsection{Mixin}

Given a view
$\view{u_\Delta}{\Gamma}{\Delta}$, an embedding $\view{u_\Phi}{\Gamma}{\Phi}$
and two disjoint injective renaming functions
$\renfuntype{\pi_\Delta}{\Delta}$ and $\renfuntype{\pi_\Phi}{\Phi}$,
where the embedding $\Phi$ decomposes as $\Phi = \substitution{\Gamma}{u_{\Phi}}{} \rtimes \Phi^+$,
we can mixin the view into the embedding,
constructing a new theory presentation $\combineResult$. We define $\combineResult \define
\combineResult_1 \rtimes \combineResult_2$ where we need a new renaming
$\pi'_{\Phi}$:
\begin{align*}
  {\pi'_{\Phi^+}} \left( y \right) & \define \left\{
  \begin{array}{ll}
    z\left[ u_{\Delta} \right]\left[ x \rename
          \pi_{\Delta}\left(x\right)\right]_{x \in \left| \Delta \right|}
       & \text{when there is a\ } z \in \left|\Gamma\right|
           \text{such that\ } z\left[u_{\Phi} \right] = y\\
    \pi_{\Phi^+} \left( y \right)  & \text{when\ } y \in \left|\Phi^+\right|
  \end{array} \right. \\
  \combineResult_1 & \define \applyrenfun{\pi_{\Delta}}{\Delta} \\
  \combineResult_2 & \define \applyrenfun{\pi'_{\Phi}}{\Phi^+} \\
\end{align*}
The mixin also provides an embedding $\view{v_{\Delta}}{\Delta}{\combineResult}$ and a view
$\view{v_{\Phi}}{\Phi}{\combineResult}$,
defined as
\begin{align*}
  \viewname{v_{\Delta}} & \define \tilde{\pi}_{\Delta}\\
  \viewname{v_\Phi} & \define \tilde{\pi}'_{\Phi} \\
\end{align*}
By definition of embedding, there is no $z\in\syms{\Gamma}$ that
is mapped into $\Phi^{+}$ by $\viewname{u_{\Phi}}$. The definition of
$\pi'_{\Phi}$ is arranged such that
$\composeviews{u_\Delta}{v_\Delta}$ is equal to $\composeviews{u_\Phi}{v_\Phi}$
(and not just equivalent);  so we can denote this joint morphism by
$\view{uv}{\Gamma}{\combineResult}$.
In other words, in a mixin, by only allowing renaming of the \emph{new components} in $\Phi^{+}$, we
insure commutativity \emph{on the nose} rather than just up to isomorphism.

Mixins also provide a set of mediating views from the constructed theory
presentation $\combineResult$. Suppose we are given the views
$\view{w_\Delta}{\Delta}{\Omega}$ and$\view{w_\Phi}{\Phi}{\Omega}$
such that the composed views
$\compose{u_\Delta}{w_\Delta}{\Gamma}{\Omega}$
and
$\compose{u_{\Phi}}{w_{\Phi}}{\Gamma}{\Omega}$ are equivalent.
We can combine $\viewname{w_{\Delta}}$ and $\viewname{w_{\Phi}}$ into
the mediating view $\view{w_\combineResult}{\combineResult}{\Omega}$ defined
as
\[
\viewdef{\viewname{w_\combineResult}}
{\substitution{}{\pi_{\Delta}(x) \mapsto \substitution{x}{w_\Delta}{}}{x \in \syms{\Delta}}
  \cup
  \substitution{}{\pi'_{\Phi}(y) \mapsto \substitution{y}{w_\Phi}{}}{y \in \syms{\Phi^+}}
}
\text{.}
\]

For mixin, again using the symbols as above, we denote the construction
results as
\[
\mixfun\left( u_{\Delta}, u_{\Phi}, \pi_{\Delta}, \pi_{\Phi}\right) \define
\left\{\begin{aligned}
  \mathtt{pres} & = \combineResult_1\rtimes\combineResult_2 \\
  \mathtt{embed}_{\Delta} & = \left[v_{\Delta}\right] : \Delta\rightarrow\combineResult\\
  \mathtt{view}_{\Phi} & = \left[v_{\Phi}\right] : \Phi\rightarrow\combineResult\\
  \mathtt{diag} & = \left[uv\right]:\Gamma\rightarrow\combineResult\\
  \mathtt{mediate} & = \lambda\ w_{\Delta}\ w_{\Phi}\ .\  w_{\combineResult}
\end{aligned}\right\}\]
Symbolically the above is very similar to what was done in combine, and indeed
we are constructing all of the data for a specific pushout.  However in this
case the results are not symmetric, as seen from the details of the construction
of $\combineResult_1$ and $\combineResult_2$, which stems from the fact that in this case
$\viewname{v_{\Phi}}$ is an arbitrary view rather than an embedding. The
\tmop{Flip} view of Section~\ref{sec:morphisms} is an example.

\section{Categorical Semantics}\label{sec:category}

We delve more deeply into the semantics of combine and mixin. 
The reason we choose category theory to do this is that the
structure we need has already largely been developed as
part of Categorical Logic and categorical semantics of
Dependent Type Theory.  In other words, because the category
of contexts is already known to have enough pullbacks and
to support suitable fibrations, it seems natural to mine
this structure for our purposes. 

The categorical interpretation (in $\catctx$) of combine is unsurprising:
pullback. But mixin is more complex: it is a Cartesian lifting in a suitable
fibration. But we also obtain that our semantics is \emph{total}, even for
mixin.  Key is that the algorithm for mixin in the previous section produces a unique
\emph{syntactic} representation of the results. Many combinators with
well-defined categorical semantics (such as unrestricted mixin) do not have
this property: while they have models, these models cannot be written down.

At first glance, the definitions of combine and mixin may appear ad hoc and
overly complicated.  This is because, in practice, the renaming functions
$\pi_{\Delta}$ and $\pi_{\Phi}$ are frequently the \emph{identity}.  The
main reason for this is that mathematical vernacular uses a lot of rigid
conventions, such as usually naming an associative, commutative, invertible
operator which possesses a unit $+$, and the unit $0$, backward composition
is $\circ$, forward composition is $;$, and so on.  But the usual notation
of lattices is different than that of semirings, even though they share a
similar ancestry -- renamings are clearly necessary at some point.

While our primary interest is in theory presentations, the bulk of the
categorical work in this area has been done on the category of contexts,
which is the opposite category.  To be consistent with the existing literature,
we will give our categorical semantics in terms of $\catctx = \catthry^{op}$.
Thus if $\arrowthry{v}{\Gamma}{\Delta}$ is a view, then a corresponding
morphism, $v$, exists
from context $\context{\Delta}$ to context $\context{\Gamma}$.
We will write such morphisms as
$\arrowctx{v}{\Delta}{\Gamma}$ when we are considering the category of contexts,
with composition as before.

\subsection{Semantics}
The category of contexts forms the base category for a fibration. The fibered
category $\catext$ is the category of context extensions.
The objects of $\catext$ are embeddings of contexts. We write such objects
as $\objectext{u}{\Gamma}{\Delta}$ where
$\theory{\Gamma}$ is the base and $\context{\Delta}$ is the 
extended context.  
The notation is to remind the reader that the morphisms are display
maps (i.e. that $\Gamma$ is a strict prefix of $\Delta$).

A morphism between two embeddings is a pair of views forming a
commutative square with the embeddings.
Thus given embeddings
$\objectext{u_2}{\Gamma_2}{\Delta_2}$
and
$\objectext{u_1}{\Gamma_1}{\Delta_1}$,
a morphism between these consists of two morphisms
$\arrowctx{v_\Delta}{\Delta_2}{\Delta_1}$ and $\arrowctx{v_\Gamma}{\Gamma_2}{\Gamma_1}$
from $\catctx$ such that
$\composeext{v_\Gamma}{u_1} = \composeext{u_2}{v_\Delta} :
\arrowctxtype{\Gamma_2}{\Delta_1}$.
When we need to be very precise, we write such a morphism as
$u_2\begin{array}{ccc}
  & v_{\Delta} & \\
  \Delta_2 & \rightarrowtriangle & \Delta_1\\
  \downharpoonright &  & \downharpoonright\\
  \Gamma_2 & \rightarrowtriangle & \Gamma_1\\
  & v_{\Gamma} &
\end{array}  u_1 $.  We will write
$\extarr{v}{\Gamma}{\Delta}{u_2}{u_1}$ whenever the rest of the information can
be inferred from context.  When given a specific morphism in $\catext$, we will
use the notation $\arrowArbExt{e}$.

A fibration of $\catext$ over $\catctx$ is defined by giving a
suitable functor from $\catext$ to $\catctx$. Our \ ``base'' functor
sends
an embedding $\objectext{e}{\Gamma}{\Delta}$ to $\Gamma$ and sends
a morphism $\extarr{v}{\Gamma}{\Delta}{u_2}{u_1}$ to
its base morphism $\arrowctx{v_\Gamma}{\Gamma_2}{\Gamma_1}$.

\begin{theorem}
  This base fibration is a Cartesian fibration.
\end{theorem}

This theorem, in slightly different form, can be found in
\cite{JacobsCLTT} and \cite{taylor1999practical}.  We give a full proof
here because we want to make the link with our mixin construction explicit.
We use the results of \S\ref{sec:combinators} directly.

\begin{proof}
  Suppose $\arrowctx{u_\Delta}{\Delta}{\Gamma}$ is a morphism in $\catctx$,
  and $\objectext{u_\Phi}{\Gamma}{\Phi}$ is an object of $\catext$ in the fiber of $\context{\Gamma}$
  (i.e. an embedding).
  We need to construct a Cartesian lifting of $\extname{u_\Delta}$, which is
  a Cartesian morphism of $\catext$ over $\extname{u_\Delta}$.
  \ The components of the mixin construction
  are exactly the ingredients we need to create this Cartesian lifting.
  Let $\renfuntype{\pi_\Delta}{\Delta}$
  and $\renfuntype{\pi'_{\Phi}}{\Phi^+}$
  be two disjoint injective renaming functions. Note that such $\pi_{\Delta}$ and
  $\pi_{\Phi}$ always exist because $\varSet$ is infinite while $\syms{\context{\Delta}}$ and
  $\syms{\context{\Phi^+}}$ are finite.
  Let
\[
\mixfun\left( \extname{u_{\Delta}}, \extname{u_{\Phi}}, \pi_{\Delta}, \pi_{\Phi}\right) \define
\left\{\begin{aligned}
  \mathtt{pres} & = \combineResult \\
  \mathtt{embed}_{\Delta} & = v_{\Delta} : \combineResult\rightarrow\Delta\\
  \mathtt{view}_{\Phi} & = v_{\Phi} : \combineResult\rightarrow\Phi\\
  \mathtt{diag} & = uv : \combineResult\rightarrow\Phi\\
  \mathtt{mediate} & = \lambda\ w_{\Delta}\ w_{\Phi}\ .\  w_{\combineResult}
\end{aligned}\right\}\]
\noindent where we recall that we are now working in $\catctx$, the opposite of
$\catthry$, and thus the direction of the morphisms is flipped.
Then $\arrowArbExt{e} \define
    v_{\Delta} \begin{array}{ccc}
    & v_{\Phi} & \\
    \combineResult & \rightarrowtriangle & \Phi\\
    \downharpoonright &  & \downharpoonright\\
    \Delta & \rightarrowtriangle & \Gamma\\
    & u_{\Delta} &
  \end{array} u_{\Phi} $ is a morphism of $\catext$ which is a
  Cartesian lift of $\extname{u_{\Delta}}$.

  Firstly, to see that $\arrowArbExt{e}$ is in fact a morphism of $\catext$, we note that
  $\arrowthry{v_\Delta}{\Delta}{\combineResult} $ is an
  embedding, so $\objectext{v_\Delta}{\Delta}{\combineResult}$
  is an object of $\catext$.
  Next we need to show that
  $\composeext{v_\Delta}{u_\Delta} = \composeext{v_\Phi}{u_\Phi}$.
  Let $z \in \syms{\context{\Gamma}}$.
  \ Then $z \left[ u_{\Delta} \right] \left[ v_{\Delta} \right] 
        = z \left[ u_{\Delta} \right]  \left[ x \rename \pi_{\Delta} \left( x \right) \right]$ 
  by definition of $\extname{v_{\Delta}}$.
  On the other hand, 
  $z \left[ u_{\Phi} \right] \left[ v_{\Phi} \right] = 
   z \left[ u_{\Phi} \right] \left[ y \mapsto {\pi'_{\Phi}} \left( y \right) \right]_{y \in \left| \Phi \right|}$
  by definition of $ v_{\Phi} $. However, $z \left[ u_{\Phi} \right]$ is a
  variable since $u_{\Phi}$ is an embedding, and by definition
  ${\pi'_{\Phi}} \left( z \left[ u_{\Phi} \right] \right) = z \left[
  u_{\Delta} \right] \left[ x \rename \pi_{\Delta} \left( x \right) \right]_{x
  \in \left| \Delta \right|}$ so that we have $z \left[ u_{\Delta} \right]
  \left[ v_{\Delta} \right] = z \left[ u_{\Phi} \right] \left[ v_{\Phi}
  \right]$ as required.

  Secondly we need to see that $\arrowArbExt{e}$ is a Cartesian lift of
  $\arrowctx{u_{\Delta}}{\Delta}{\Gamma}$.
  We need to show that for any
  morphism $\arrowArbExt{f} \define w_{\Psi}  \begin{array}{ccc}
    & w_{\Phi}  & \\
    \Omega & \rightarrowtriangle & \Phi\\
    \downharpoonright &  & \downharpoonright\\
    \Psi & \rightarrowtriangle & \Gamma\\
    & w_{\Gamma}  &
  \end{array} u_{\Phi} $ from $\catext$ and any arrow
  $\arrowctx{w_0}{\Psi}{\Delta}$ from $\catctx$ such that
  $\extname{w_\Gamma} = \extname{w_0} ; \extname{u_{\Delta}} :
  \arrowctxtype{\Psi}{\Gamma}$, there is a unique mediating morphism
  $\arrowArbExt{h} \define 
  w_{\Psi}  \begin{array}{ccc}
    & w_{\combineResult}  & \\
    \Omega & \rightarrowtriangle & \combineResult\\
    \downharpoonright &  & \downharpoonright\\
    \Psi & \rightarrowtriangle & \Delta\\
    & w_0  &
  \end{array} v_{\Delta} $ from $\catext$ such that
\begin{equation}
\arrowArbExt{h}\ ;\ \arrowArbExt{e} = \arrowArbExt{f}
  \label{EcomposeArrow}
\end{equation}
  % -----------------
  To show that such an $\arrowArbExt{h}$ exists, we only need to construct
  $\arrowctx{w_{\combineResult}}{\Omega}{\combineResult}$ and show that it
  has the
  required properties. We will show that the mediating morphism $\extname{w}$ from the mixin
  construction given
  $\arrowctx{w_\Phi}{\Omega}{\Phi}$
  and
  $w_{\Delta}  \define w_{\Psi}
   ; w_0  : \arrowctxtype{\Omega}{\Delta}$ is the
  required morphism.

  First we note that $w_{\Delta} ; u_{\Delta} = w_{\Phi} ; u_{\Phi} $ as
  required by the mixin construction for the mediating morphism since 
  $ w_{\Delta} ; u_{\Delta} = 
    w_{\Psi} ; w_0 ; u_{\Delta} =
    w_{\combineResult} ; v_{\Phi}; u_{\Phi} =
    w_{\Phi} ; u_{\Phi} $ by chasing around the diagram of the equality
  $\arrowArbExt{h}\ ;\ \arrowArbExt{e} = \arrowArbExt{f}$.
  Now taking $w_{\combineResult}  \define w $ we
  need to show that $\arrowArbExt{h}$
  is a well defined morphism in
  $\catext$ by showing it forms a commutative square. Suppose $x \in
  \left| \Delta \right|$. Then 
  $x \left[ v_{\Delta}\right] \left[ w_{\combineResult}
  \right] = \pi_{\Delta} \left( x \right) \left[ w_{\combineResult} \right] = x \left[
  w_{\Delta} \right] = x \left[ w_0 \right] \left[ w_{\Psi} \right]$ as
  required. Next we need to show that equation~(\ref{EcomposeArrow}) holds. It
  suffices to show that 
  $w_{\combineResult} ; v_{\Phi} = w_{\Phi} $ since it is already required that 
  $w_0 ; u_{\Delta} = w_{\Gamma}$.
  Suppose $y \in \left| \Phi \right|$. There are two possiblities, either $y =
  z \left[ u_{\Phi} \right]$ for some $z \in \left| \Gamma \right|$, or $y \in
  \left| \Phi^+ \right|$ where $\Phi = \Gamma \left[ u_{\Phi} \right] \rtimes
  \Phi^+$. \ If $y \in \left| \Phi^+ \right|$ then $y \left[ v_{\Phi} \right]
  \left[ w_{\combineResult} \right] = \pi_{\Phi^+} \left( y \right) \left[ w_{\combineResult}
  \right] = y \left[ w_{\Phi} \right]$ as requried. \ In case $y = z \left[
  u_{\Phi} \right]$, then $y \left[ v_{\Phi} \right] \left[ w_{\combineResult} \right] =
  z \left[ u_{\Phi} \right] \left[ v_{\Phi} \right] \left[ w_{\combineResult} \right] = z
  \left[ u_{\Delta} \right] \left[ v_{\Delta} \right] \left[ w_{\combineResult} \right] =
  z \left[ u_{\Delta} \right] \left[ w_0 \right] \left[ w_{\Psi} \right] = z
  \left[ w_{\Gamma} \right] \left[ w_{\Psi} \right] = z \left[ u_{\Phi}
  \right] \left[ w_{\Phi} \right] = y \left[ w_{\Phi} \right]$ as requiried.

  Lastly we need to show that the mediating morphism $\arrowArbExt{h}$
  is the unique morphism satisfying
  equation~(\ref{EcomposeArrow}).  Let $\arrowArbExt{j}$ be another
  morphism of $\E$, where $\arrowArbExt{j}$ must have the same shape
  as $\arrowArbExt{h}$, but with $w_{\combineResult}$ replaced
  with $w'_{\combineResult}$.  Suppose that
\[ \arrowArbExt{j}\ ;\ \arrowArbExt{f} = \arrowArbExt{e} \]
  We need to show that $ w'_{\combineResult}  =  w_{\combineResult} $. \
  Suppose $z \in \left| \combineResult \right|$. \ There are two possiblities. \ Either
  $z = x \left[ v_{\Delta} \right]$ for some $x \in \left| \Delta \right|$ or
  $z = y \left[ v_{\Phi} \right]$ for some $y \in \left| \Phi^+ \right|$. \
  Suppose $z = x \left[ v_{\Delta} \right]$. Then $z \left[ w'_{\combineResult} \right] =
  x \left[ v_{\Delta} \right] \left[ w'_{\combineResult} \right] = x \left[ w_0 \right]
  \left[ w_{\Psi} \right] = x \left[ v_{\Delta} \right] \left[ w_{\combineResult} \right]
  = z \left[ w_{\combineResult} \right]$ as required. \ On the other hand, suppose $z = y
  \left[ v_{\Phi} \right]$. Then $z \left[ w'_{\combineResult} \right] = y \left[
  v_{\Phi} \right] \left[ w'_{\combineResult} \right] = y \left[ w_{\Phi} \right] = y
  \left[ v_{\Phi} \right] \left[ w_{\combineResult} \right] = z \left[ w_{\combineResult} \right]$
  as required. So $ w'_{\combineResult} =  w_{\combineResult}$ and
  hence $\arrowArbExt{j} = \arrowArbExt{h}$, as required. \qed
\end{proof}

The above proof illustrates that the mixin operation is characterized by the
properties of a Cartesian lifting in the fibration of embeddings. Notice that a
Cartesian lift is only characterised up to isomorphism. Thus there are
potentially many isomorphic choices for a Cartesian lift, and hence
there are many possible choices for how to mixin an embedding into a view.
This is the underlying reason why the mixin construction requires a pair of
renaming functions. The renaming functions pick out a particular choice of
mixin from the many possibilities. This ability to specify which mixin
to construct is quite important as one cannot simply define a mixin to
be ``the'' Cartesian lift, since ``the'' Cartesian lift is only defined up to
isomorphism.  It is important to remember that for \emph{user syntax}, we
cannot work up to isomorphism!

Next we will see that combine is a special case of mixin.

\begin{theorem}\label{thm:combine}
  Given two embeddings $ u_{\Delta}  : \extend{\Gamma}{\Delta}$
  and $ u_{\Phi}  : \extend{\Gamma}{\Phi}$ and renaming
  functions $\pi_{\Delta} : \left| \Delta \right| \rightarrow \varSet$ and
  $\pi_{\Phi} : \left| \Phi \right| \rightarrow \varSet$ sastifiying the
  requirement of the combine construction, then
\begin{equation}\label{mixcombinerel}
\mixfun\left( u_{\Delta}, u_{\Phi}, \pi_{\Delta}, \pi_{\Phi^{+}}\right) =
\comfun\left( u_{\Delta}, u_{\Phi}, \pi_{\Delta}, \pi_{\Phi} \right)
\end{equation}
\noindent where $\Phi = \Gamma\left[u_{\Phi}\right]\rtimes\Phi^{+}$ and
$\pi_{\Phi{^{+}}} = [x \mapsto \pi_{\Phi}]_{x\in|\Phi^{+}|}$,
and equation~\ref{mixcombinerel} is interpreted component-wise.
\end{theorem}

\begin{proof}
  Suppose that
\[
\comfun\left( u_{\Delta}, u_{\Phi}, \pi_{\Delta}, \pi_{\Phi}\right) =
\left\{\begin{aligned}
  \mathtt{pres} & = \combineResult_0\rtimes\left(\combineResult_{\Delta} \cup \combineResult_{\Phi}\right) \\
  \mathtt{embed}_{\Delta} & = v_{\Delta} : \combineResult\rightarrow\Delta\\
  \mathtt{embed}_{\Phi} & = v_{\Phi} : \combineResult\rightarrow\Phi\\
  \mathtt{diag} & = uv : \combineResult\rightarrow\Gamma\\
  \mathtt{mediate} & = \lambda\ w_{\Delta}\ w_{\Phi}\ .\  w_{\combineResult}
\end{aligned}\right\}\]
and
\[
\mixfun\left( u_{\Delta}, u_{\Phi}, \pi_{\Delta}, \pi_{\Phi^{+}}\right) =
\left\{\begin{aligned}
  \mathtt{pres} & = \combineResult' \\
  \mathtt{embed}_{\Delta} & = v'_{\Delta} : \combineResult'\rightarrow\Delta\\
  \mathtt{view}_{\Phi} & = v'_{\Phi} : \combineResult'\rightarrow\Phi\\
  \mathtt{diag} & = uv' : \combineResult'\rightarrow\Gamma\\
  \mathtt{mediate} & = \lambda\ w_{\Delta}\ w_{\Phi}\ .\  w_{\combineResult'}
\end{aligned}\right\}\]
  Recall that $\combineResult = \combineResult_0 \rtimes \left( \combineResult_{\Delta} \cup \combineResult_{\Phi} \right)
  = \combineResult_0 \rtimes \left(\combineResult_{\Delta} \right) \rtimes \combineResult_{\Phi}$ where\\
  $\combineResult_0
  \define \Gamma \left[ z \rename \pi_{\Delta} \left( z \left[ v_{\Delta}
  \right] \right) \right]_{z \in \left| \Gamma \right|}$, $\combineResult_{\Delta}
  \define \Delta^+ \left[ x \rename \pi_{\Delta} \left( x \right) \right]_{x
  \in \left| \Delta \right|}$, and $\combineResult_{\Phi} \define \Phi^+ \left[ y \rename
  \pi_{\Phi} \left( y \right) \right]_{y \in \left| \Phi \right|}$. In
  particular note that $\combineResult_0 = \Gamma \left[ v_{\Delta} \right] \left[ z
  \left[ v_{\Delta} \right] \rename \pi_{\Delta} \left( z \left[ v_{\Delta}
  \right] \right) \right]_{z \in \left| \Gamma \right|}$. \ Since $\Delta =
  \Gamma \left[ v_{\Delta} \right] \rtimes \Delta^+$, we have that
  \begin{eqnarray*}
    \combineResult_0 \rtimes \combineResult_{\Delta} & = & \Gamma \left[ v_{\Delta} \right] \left[ z
    \left[ v_{\Delta} \right] \rename \pi_{\Delta} \left( z \left[ v_{\Delta}
    \right] \right) \right]_{z \in \left| \Gamma \right|} \rtimes \Delta^+
    \left[ x \rename \pi_{\Delta} \left( x \right) \right]_{x \in \left|
    \Delta \right|}\\
    & = & \left( \Gamma \left[ v_{\Delta} \right] \rtimes \Delta^+ \right)
    \left[ x \rename \pi_{\Delta} \left( x \right) \right]_{x \in \left|
    \Delta \right|}\\
    & = & \Delta \left[ x \rename \pi_{\Delta} \left( x \right) \right]_{x
    \in \left| \Delta \right|}
  \end{eqnarray*}
  Recall also that $\combineResult' = \combineResult'_1 \rtimes \combineResult'_2$ where $\combineResult_1' \define \Delta
  \left[ x \rename \pi_{\Delta} \left( x \right) \right]_{x \in \left| \Delta
  \right|}$ and $\combineResult'_2 \define \Phi^{\upl} \left[ y \rename
  {\pi'_{\Phi^{\upl}}} \left( y \right) \right]_{y \in \left| \Phi
  \right|}$. So we see that $\combineResult_1' = \combineResult_0 \rtimes \combineResult_{\Delta}$.

  Next we show that ${\pi'_{\Phi}} = \pi_{\Phi}$. If $y \in
  \left| \Phi \right|$ then either $y \in \left| \Phi^+ \right|$ or there is
  some $z \in \Gamma$ such that $y = z \left[ v_{\Phi} \right]$. If $y \in
  \left| \Phi^+ \right|$ then ${\pi'_{\Phi}} \left( y \right) =
  \pi_{\Phi} \left( y \right) = \pi_{\Phi} \left( y \right)$. If $y = z
  \left[ v_{\Phi} \right]$, then ${\pi'_{\Phi}} \left( y \right)
  = z \left[ u_{\Delta} \right] \left[ x \rename \pi_{\Delta} \left( x \right)
  \right]_{x \in \left| \Delta \right|} = \pi_{\Delta} \left( z \left[
  u_{\Delta} \right] \right) = \pi_{\Phi} \left( z \left[ u_{\Phi} \right]
  \right) = \pi_{\Phi} \left( y \right)$. Therefore $\combineResult_2' = \combineResult_{\Phi}$ and
  hence $\combineResult' = \combineResult$.

  Next we need to show that $ v'_{\Delta}  =  v_{\Delta}
  $ and $ v'_{\Phi}  =  v_{\Phi} $. First we
  see that $ v'_{\Delta} $ and $ v_{\Delta} $ are
  both defined to be $\left[ x \rename \pi_{\Delta} \left( x \right)
  \right]_{x \in \left| \Delta \right|}$, so clearly they are equal. \ Next we
  see that $v_{\Phi} \define \left[ y \rename \pi_{\Phi} \left(
  y \right) \right]_{y \in \left| \Phi \right|}$ and $ v'_{\Phi} 
  \define \left[ y \rename {\pi'_{\Phi}} \left( y \right) \right]_{y
  \in \left| \Phi \right|}$ are equal because ${\pi'_{\Phi}} =
  \pi_{\Phi}$.  This also gives that $uv = uv'$.

  Lastly we show that the mediating morphism of the combine is the same as the
  mediating morphism of the mixin. \ Suppose we are given $ w_{\Delta}
   : \Delta \rightarrow \Omega$ and $ w_{\Phi}  : \Phi
  \rightarrow \Omega$ such that $ u_{\Delta}  ;  w_{\Delta}
   =  u_{\Phi}  ;  w_{\Phi}  : \Gamma
  \rightarrow \Omega$. To show that the mediating morphism produced by combine,
  $ w_{\combineResult}  : \combineResult \rightarrow \Omega$ is the same as the
  medating morphism produced by the mixin, it suffices to prove that the mediating
  morphism satifies the universal property of the Cartesian lift, since such a
  morphism is unique. Thus it suffices to show that $ v_{\Phi}  ;
   w_{\combineResult}  =  w_{\Phi}  : \Phi \rightarrow \Omega$
  and $ v_{\Delta}  ;  w_{\combineResult}  =  w_{\Delta}
   : \Delta \rightarrow \Omega$. Let $y \in \left| \Phi \right|$. \
  Then $y \left[ v_{\Phi} \right] \left[ w_{\combineResult} \right] = \pi_{\Phi} \left( y
  \right) \left[ w_{\combineResult} \right] = y \left[ w_{\Phi} \right]$. Let $x \in
  \left| \Delta \right|$. Then $x \left[ v_{\Delta} \right] \left[ w_{\combineResult}
  \right] = \pi_{\Delta} \left( x \right) \left[ w_{\combineResult} \right] = x \left[
  w_{\Delta} \right]$ as required.
\end{proof}

Combine is rather well-behaved.  In particular,
\begin{proposition}\label{combine-commutative}
$\comfun\left( u_{\Delta}, u_{\Phi}, \pi_{\Delta}, \pi_{\Phi}\right) =
\comfun\left( u_{\Phi}, u_{\Delta}, \pi_{\Phi}, \pi_{\Delta}\right)$,
i.e.  combine is commutative.
\end{proposition}

It turns out that combine also satisfies an appropriate notion of
associativity.  In other words, we can compute limits of cones of
embeddings.

\subsection{No Lifting Views over Views}

Why do we restrict ourselves to the fibration of embeddings? Why not allow
mixins of arbitrary views over arbitrary views? If such mixins
were allowed, then the notion of a Cartesian lifting reduces to that of
pullback. But to demand that the category of contexts and views be closed under
\emph{all} pullbacks would require too much from our type theory: we would
need to have all equalizers (as we already have all products).  In particular,
at the type level, this would force us to have subset types, which is something
we are not willing to impose.  Thus a restriction is
needed, and our proposed restriction of only mixing in embeddings into
views appears to be quite practical.  Taylor~\cite{taylor1999practical} is
a good source of further reasons for the naturality of restricting to this
case.

\section{Presentation Combinators}
\label{sec:tpc2}

We can now reconstruct a syntax for our theory presentation combinators,
based on the semantics of the previous sections.  Rather than attempt
to patch our previous syntax, we directly use the semantics to guide us.
In fact, \emph{guide} is the wrong word: we let the semantics induce
the syntax, instead of letting our na\"{\i}ve intuition about what
combinators might be useful dictate the syntax.

\subsection{Grammar}
\label{sec:grammar}

In the definition of the grammar,
we use $A,B$ to denote theories and/or views,
$x$ and $y$ to denote symbols, $t$ for terms of the
underlying type theory, and $l$ for (raw) contexts from the
underlying type theory.

\begin{minipage}[t]{0.45\textwidth}
\begin{align*}
\mathrm{tpc} \Coloneqq
   &\  \mathsf{Empty} \\
 | &\  \mathsf{Theory\ }\{ l \} \\
 | &\ \mathsf{extend}\  A\  \mathsf{by}\  \{ l \} \\
 | &\  \mathsf{combine}\  A\  r_1,\  B\  r_2 \\
 | &\ \mathsf{mixin}\ A\ r_1,\ B\ r_2\\
 | &\ \mathsf{view}\ A\ \mathsf{as}\ B\ \mathsf{via}\ \mathrm{v} \\
 | &\  A\ ;\ B \\
 | &\  A\ \mathrm{r}
\end{align*}
\end{minipage}
\begin{minipage}[t]{0.45\textwidth}
\begin{align*}
\mathrm{r} \Coloneqq &\ \left[ \mathrm{ren} \right] \\
\mathrm{v} \Coloneqq &\ \left[ \mathrm{assign} \right] \\
\mathrm{ren} \Coloneqq &\ x \mapsto y \\
 | &\ \mathrm{ren},\ x \mapsto y \\
\mathrm{assign} \Coloneqq &\ x \mapsto t \\
 | &\ \mathrm{assign},\ x \mapsto t
\end{align*}
\vfill
\end{minipage}

Informally, these forms correspond to the empty theory, an explicit
theory, a theory extension, combining two extensions, mixing in a
view and an extension, explicit view, sequencing views, and renaming.

What might be surprising is that we do not have a separate language for
presentations and views.  This is because our language does not have
a single semantics in terms of presentations, embeddings or views, but rather
has \emph{several} compatible semantics.  In other words, our syntax will
yield objects of $\catctx$, objects of $\catext$ (i.e. embeddings) and
morphisms of $\catctx$ (views).

The semantics is given by defining three partial maps,
$\semC{-} :\ \mathsf{tpc} \partialf |\catctx|$,
$\semE{-} :\ \mathsf{tpc} \partialf |\E|$,
$\semB{-} :\ \mathsf{tpc} \partialf \mathrm{Hom}_{\catctx}$.  This is done
by simultaneous structural recusion.  We also use $\semP{-}$ for the
straightforward semantics in $\V\rightarrow\V$ of a renaming.

\begin{align*}
\semC{-} :\ & \mathsf{tpc} \partialf |\catctx| \\
\semC{\mathsf{Empty}} =\  & \EmptyThy\\
\semC{\mathsf{Theory\ }\{ l \}} =\ & l \qquad\qquad \text{when}\ l\ \wfctx\\
\semC{\mathsf{extend\ }A\mathsf{\ by\ }\{ l \}} =\  &
    \extension{\semC{A}}{l}.\mathtt{pres} \\
\semC{\mathsf{combine\ } A_1 r_{1},\ A_2 r_{2}} =\  &
    \comfun\left(\semE{A_1}, \semE{A_2}, \semP{r_1}, \semP{r_2}\right).\mathtt{pres} \\
\semC{\mathsf{mixin\ } A_1 r_{1},\ A_2 r_{2}} =\  &
    \mixfun\left(\semB{A_1}, \semE{A_2}, \semP{r_1}, \semP{r_2}\right).\mathtt{pres} \\
\semC{\mathsf{view}\ A\ \mathsf{as}\ B\ \mathsf{via}\ \mathrm{v}}
    =\ & \bot \\
\semC{A ; B} =\  & \mathsf{cod}\ \semB{A ; B} \\
\semC{A\ r} =\  & \renameOp{\semC{A}}{\semP{r}}.\mathtt{pres} \\
\end{align*}

Recall that objects of $\E$ corresponds to those morphisms of $\catctx$ (i.e.
views) which are in fact embeddings.

\begin{align*}
\semE{-} :\ & \mathsf{tpc} \partialf |\E| \\
\semE{\mathsf{Empty}} =\  &
    \mathbb{I}_{\EmptyThy}\\
\semE{\mathsf{Theory\ }\{ l \}} =\ &
  !_{l} : \emptyview\rightarrow \semC{l}\\
\semE{\mathsf{extend\ }A\mathsf{\ by\ }\{ l \}} =\  &
    \extension{\semC{A}}{l}.\mathtt{embed} \\
\semE{\mathsf{combine\ } A_1 r_{1}, A_2 r_{2}} =\  &
    \comfun\left(\semE{A_1}, \semE{A_2}, \semP{r_1}, \semP{r_2}\right).\mathtt{diag} \\
\semE{\mathsf{mixin\ } A_1 r_{1}, A_2 r_{2}} =\  &
    \bot \\
\semE{\mathsf{view}\ A\ \mathsf{as}\ B\ \mathsf{via}\ \mathrm{v}}
    =\ & \bot \\
\semE{A ; B} =\  & \semE{A} ; \semE{B} \\
\semE{A\ r} =\  & \renameOp{\semC{A}}{\semP{r}}.\mathtt{embed} \\
\end{align*}

Lastly, morphisms of $\catctx$ are views.
\begin{align*}
\semB{-} :\ & \mathsf{tpc} \partialf \mathrm{Hom}_{\catctx} \\
\semB{\mathsf{Empty}} =\  &
    \mathbb{I}_{\EmptyThy}\\
\semB{\mathsf{Theory\ }\{ l \}} =\ &
  !_{l} : \emptyview\rightarrow \semC{l}\\
\semB{\mathsf{extend\ }A\mathsf{\ by\ }\{ l \}} =\  &
    \extension{\semC{A}}{l}.\mathtt{embed} \\
\semB{\mathsf{combine\ } A_1 r_{1}, A_2 r_{2}} =\  &
    \comfun\left(\semE{A_1}, \semE{A_2}, \semP{r_1}, \semP{r_2}\right).\mathtt{diag} \\
\semB{\mathsf{mixin\ } A_1 r_{1}, A_2 r_{2}} =\  &
    \comfun\left(\semB{A_1}, \semE{A_2}, \semP{r_1}, \semP{r_2}\right).\mathtt{diag} \\
\semB{\mathsf{view}\ A\ \mathsf{as}\ B\ \mathsf{via}\ \mathrm{v}}
    =\ & [v] : \semC{A} \rightarrow \semC{B} \\
\semB{A ; B} =\  & \semB{A} ; \semB{B} \\
\semB{A\ r} =\  & \renameOp{\semC{A}}{\semP{r}}.\mathtt{embed} \\
\end{align*}

All rules are strictly compositional except for $\semC{A;B}$, but this is
ok since $\semB{A;B}$ is compositional.

We thus get $3$ \emph{elaborators}, as the members of $|\catctx|$,
$|\E|$ and $|\mathrm{Hom}_{\catctx}|$ can all be represented syntactically
in the underlying type theory.

Note that we could have interpreted
$\semC{\mathsf{view}\ A\ \mathsf{as}\ B\ \mathsf{via}\ \mathrm{v}}$ as
$\mathsf{cod} \semB{\mathsf{view}\ A\ \mathsf{as}\ B\ \mathsf{via}\ \mathrm{v}}$,
rather than as $\bot$, but this is not actually helpful, since this is
just $\semC{B}$,  which is not actually what we want.  What we would really
want is the result of doing the substitution $v$ into $A$, but the resulting
presentation may no longer be well-formed.
So we chose to interpret the attempt to take the object
component of a view as a specification error.  Similarly, even though we can
give an interpretation as an embedding for mixin when $A_1$ turns out
to be an embedding, and also for an embedding $r$ in a view context
(i.e.  $\mathsf{view}\ A\ \mathsf{as}\ B\ \mathsf{via}\ \mathrm{r}$),
we also choose to make these specification errors as well.

We should also note here that in our implementation, we allow raw renamings
($[\mathrm{ren}]$) and assignments ($[\mathrm{assign}]$) to be named, for easier
reuse.  While renamings can be given a simple categorical semantics (they
induce a natural transformation on $\catctx$), assignments really need to be
interpreted contextually since this requires checking that terms $t$ are
well-typed.

Furthermore, we add a bit of syntactic sugar: $A || B$ stands for
$\mathsf{combine\ }A\ [], B\ []$, a rather common situation.

\subsection{Type System}
\label{sec:types}

\newcommand{\thryty}{\tmop{Th}}
\newcommand{\extty}[2]{\tmop{Emb\ #1\ #2}}
\newcommand{\viewty}[2]{\tmop{View\ #1\ #2}}
\newcommand{\permty}[1]{\tmop{Perm\ #1}}
\newcommand{\assignty}[2]{\tmop{Assign\ #1\ #2}}
\newcommand{\toctx}[1]{\tmop{#1_{ctx}}}
\newcommand{\toctxTy}{\tmop{C\ :\ Th \to Ctx}}
\renewcommand{\assignty}[2]{\tmop{Assign\ #1 \ #2}}
\newcommand{\combine}{\tmop{combine}\ A\ r_1\ B\ r_2}
\newcommand{\mixin}{\tmop{mixin}\ A\ r_1\ B\ r_2}

\renewcommand{\empty}{\tmop{Empty}}
\newcommand{\invariant}{
  \forall x \in |A|, y \in |B| .
  r_1 \left( x \right) = r_2 \left( y \right)
  \Leftrightarrow \exists z \in \left| C \right|.\ x =
  \substitution{z}{A}{} \wedge y = \substitution{z}{B}{}}
\newcommand{\combineRes}{(C[A] \rtimes r_1 \cdot A^+ \rtimes r_2 \cdot B^+)}

\begin{figure}[th]
\begin{proofrules}
\[
\emptyset \subseteq T 
\justifies 
[] : \lstinline|Perm T|
\]  
\[
P : \lstinline|Perm S| \qquad 
S \cap \{x,y\} = \emptyset \qquad S \cup \{x,y\} \subset T
\justifies 
P,[x \leftrightarrow y] : \lstinline|Perm T|
\]
\end{proofrules}  

\begin{proofrules}
\[
A : \thryty \qquad
B : \thryty \qquad 
\toctx{B} \vdash x_i : \sigma_i[x_1,\cdots,x_{i-1}]\qquad \linebreak 
\toctx{A} \vdash y_i : \sigma_i[y_1/x_1,\cdots,y_{i-1}/x_{i-1}]
\justifies
[x_i \mapsto y_i]_{i\leq n} : \assignty{A}{B}
\]
\end{proofrules}

\caption{Types for permutations and assignments}
\label{fig:permandassign}
\end{figure}

We build a type system, whose purpose is to insure that well-typed
expressions are denoting.  Note that although we requite that
the underlying system have kinds, to enable the declaration of new
types, we omit this below for clarity. Adding this is straightforward.

First, we need a couple of preliminary
types, shown in Figure~\ref{fig:permandassign}.  We use $S$ and $T$ to
denote finite sets of variable from $\vars$, and $\toctx{A}$ means
the well-formed context that theory $A$ elaborates to.  The rule
for assignments is otherwise the standard one for morphisms of
the category of contexts (p.602 of~\cite{JacobsCLTT}).

\begin{figure}
\begin{proofrules}
% empty theory   
\[
\ \justifies 
\empty : \extty{\emptyset}{\emptyset}
\]
% list of declarations 
\[
\vdash \ell\ \wfctx
\justifies 
\tmop{Theory\ \{\ell\}} : \extty{\emptyset}{\ell}
\]
% rename 
\[
A : \thryty \qquad 
r : \permty{\emptyset}
\justifies
A\ r : \extty{A}{(A\ r)}
\]
\\ [5ex]
% extension with an empty list 
\[
A : \thryty
\justifies 
\tmop{extend\ A\ by\ \{\}} : \extty{A}{A}
\]
% extension with a list of declarations 
\[
B = \tmop{extend\ A\ by\ \{\ell\}}  \qquad B :\thryty \qquad
x \notin \syms{B} \qquad 
\toctx{B} \vdash t : type 
\justifies 
\tmop{extend\ A\ by\ \{\ell,\ x\ :\ t\}} : \extty{A}{(B,x:t)}
\]
% combine 
\\ [5ex]
\infer{\combine :  \extty{C}{\combineRes}}{
  \deduce{
    A : \extty{C}{A^\prime} \qquad
    B : \extty{C}{B^\prime} \qquad
    r_1, r_2 : \permty{\emptyset} \qquad 
    \invariant}{
    C, A^\prime, B^\prime : \thryty \qquad 
    A^\prime = C[A] \rtimes A^+ \qquad 
    B^\prime = C[B] \rtimes B^+ \qquad }
}
\end{proofrules}
\caption{Types for core combinators}
\label{fig:corecomb}
\end{figure}

We need to define $3$ sets of typing rules, one for each
semantic. The rules are extremely similar to each other for most
of the combinators, and thus we give a lighter presentation
by grouping the similar ones together.  More specifically,
we introduce judgement for $3$ new types: $\thryty$ for
theory presentations, $\extty{A}{B}$ for embeddings
from presentation $A$ to presentation $B$, and
$\viewty{A}{B}$ for views from presentation $A$ to
presentation $B$.

Figure~\ref{fig:corecomb} shows the judgements for $\extty{A}{B}$ for
all combinators except for \tmop{mixin}, \tmop{view} and sequential composition
\tmop{;}.  The judgements for $\thryty$ are obtained from these by
replacing the final $: \extty{A}{B}$ with $: \thryty$. The judgements
are defined by mutual recursion: a \tmop{combine}, even elaborated as
a theory, does take two views as arguments. The judgements for
$\extty{A}{B}$ are obtained by replacing the final $: \extty{A}{B}$ with
$: \viewty{A}{B}$ (recall that all embeddings are views).

\begin{figure}[th]
\begin{proofrules}
\[
A, B, C : \thryty \qquad
X : \viewty{A}{B} \qquad
Y : \viewty{B}{C} \qquad
\justifies
X ; Y : \thryty
\]
\[
A, B, C : \thryty \qquad
X : \extty{A}{B} \qquad
Y : \extty{B}{C} \qquad
\justifies
X ; Y : \extty{A}{C}
\]
\[
A, B, C : \thryty \qquad
X : \viewty{A}{B} \qquad
Y : \viewty{B}{C} \qquad
\justifies
X ; Y : \viewty{A}{C}
\]
\end{proofrules}
\caption{Types for composition ($;$)}
\label{fig:compo}
\end{figure}

Figure~\ref{fig:compo} shows the $3$ judgements for $;$. These
are also mutually recursive, but in a non-uniform pattern. Also,
in light of Proposition~\ref{emb-are-views} below, the first and
third rules really stand for $4$ rules each.

\begin{figure}[th]
\begin{proofrules}
\infer{\mixin : \thryty}{
  \deduce{
    A : \viewty{C}{A^\prime} \qquad
    B : \extty{C}{B^\prime} \qquad
    r_1, r_2 : \permty{\emptyset} \qquad 
    \invariant}{
    C, A^\prime, B^\prime : \thryty \qquad 
    A^\prime = C[A] \rtimes A^+ \qquad 
    B^\prime = C[B] \rtimes B^+ \qquad }
}
\infer{\mixin :  \viewty{C}{\combineRes}}{
  \deduce{
    A : \viewty{C}{A^\prime} \qquad
    B : \extty{C}{B^\prime} \qquad
    r_1, r_2 : \permty{\emptyset} \qquad 
    \invariant}{
    C, A^\prime, B^\prime : \thryty \qquad 
    A^\prime = C[A] \rtimes A^+ \qquad 
    B^\prime = C[B] \rtimes B^+ \qquad }
}
\[
A : \thryty \qquad 
B : \thryty \qquad 
v : \assignty{A}{B}
\justifies 
\tmop{view}\ A\ \tmop{as}\ B\ \tmop{as}\ v : \viewty{A}{B}
\]
\end{proofrules}
\caption{Types for mixin and view}
\label{fig:mixin-view}
\end{figure}

The rules for \tmop{mixin} (in Figure~\ref{fig:mixin-view}) are
very similar to those for \tmop{combine}, except the the first argument
$A$ is now a view, and the result is either a presentation or a view.
A $\tmop{view}$ of course just elaborates to a view, but from an
assignment.

Of course, we then have a basic soundness result:
\begin{theorem} The following hold: \\
\begin{itemize}
\item If $C : \thryty$ then $\semC{C}$ is defined.
\item If $C : \extty{A}{B}$ then $\semE{C}$ is defined.
\item If $C : \viewty{A}{B}$ then $\semB{C}$ is defined.
\end{itemize}
\end{theorem}
As can be seen both in the algorithms and in the rules above,
the interpretation as a presentation of an embedding or a view
are the target theories themselves:
\begin{proposition}
If $C : \extty{A}{B}$ then $\semC{C} = \semC{B}$.
\end{proposition}
\begin{proposition}
If $C : \viewty{A}{B}$ then $\semC{C} = \semC{B}$.
\end{proposition}
\begin{proposition}\label{emb-are-views}
If $C : \extty{A}{B}$ then $C : \viewty{A}{B}$.
\end{proposition}

\section{Examples}\label{sec:examples}

We show some progressively more complex examples, drawn from our 
library~\cite{TPCProto}.
These are chosen to illustrate the power of the combinators, and how
they solve the various problems we highlighted in \S\ref{sec:motivation}.

The simplest use of $\mathsf{combine}$ comes very quickly in a hierarchy
built using \emph{tiny theories}, namely when we construct a
pointed magma from a magma and (the theory of) a point.

\begin{lstlisting}[mathescape]
Carrier := Empty   extended_by { U : type }
Magma   := Carrier extended_by { * : U $\to$ U $\to$ U }
Pointed := Carrier extended_by { e : U }
PointedMagma := Magma || Pointed
\end{lstlisting}
\noindent where we have used the $||$ sugar for $\mathsf{combine}$.  The
result is an embedding
$\semE{\tmop{PointedMagma}} : \semC{\tmop{Carrier}} \rightarrow \semC{\tmop{PointedMagma}}$.

\noindent If we want a theory of \emph{two} points, we need to rename one of
them:
\begin{lstlisting}[mathescape]
TwoPointed := combine Pointed [], Pointed [e $\mapsto$ e$^\prime$]
\end{lstlisting}
\noindent  We can also extend by properties:
\begin{lstlisting}[mathescape]
LeftUnital := PointedMagma extended_by {
    axiom leftIdentity : $\forall$ x:U. e * x = x
}
\end{lstlisting}

This illustrates a design principle: properties should be defined as
extensions of their minimal theory.  Such minimal theories are most often
\emph{signatures}, in other words property-free theories.  By the results of
the previous section, this maximizes reusability.  Even though signatures
have no specific status in our framework, they arise very naturally as
``universal base points'' for theory development.

$\tmop{LeftUnital}$ has a natural dual, $\tmop{RightUnital}$.
While $\tmop{RightUnital}$ is straightforward to define explicitly, this should nevertheless
give pause, as this is really duplicating information which already exists.
We can use the following self-view to capture that information:
\begin{lstlisting}[mathescape]
Flip := view Magma as Magma via [ * $\mapsto$ $\lambda$ x y $\cdot$ y * x ]
\end{lstlisting}
\noindent Note that there is no interpretation for $\semC{\tmop{Flip}}$
as a theory or as an embedding;
if we were to perform the substitution directly, we would obtain
\begin{lstlisting}
Theory { U : type; fun (x,y). y * x : (U,U) -> U }
\end{lstlisting}
which is ill-defined since it has a non-symbol on the left-hand-side, and
it contains the undefined symbol $*$.

One could be tempted to write
\begin{lstlisting}
RightUnital := mixin Flip [], LeftUnital []
\end{lstlisting}
\noindent but this is also incorrect since $\tmop{LeftUnital}$ is an
extension from $\tmop{PointedMagma}$, not $\tmop{Magma}$.  The solution
is to write
\begin{lstlisting}
RightUnital := mixin Flip [], (PointedMagma ; LeftUnital) []
\end{lstlisting}
\noindent which gives a correct result, but with an axiom still called
\tmop{leftIdentity}; the better solution is to write
\begin{lstlisting}[mathescape]
RightUnital := mixin Flip [],
 (PointedMagma ; LeftUnital)[leftIdentity $\mapsto$ rightIdentity]
\end{lstlisting}
\noindent which is the \tmop{RightUnital} we want.  The construction
also makes available an embedding from \tmop{Magma} (as if we had done the
construction manually) as well as views from \tmop{LeftUnital} and
\tmop{Magma}.

The syntax used above is sub-optimal: the path \tmop{PointedMagma ;
LeftUnital} may well be needed again, and should be named.  In other words,
\begin{lstlisting}
LeftUnit := PointedMagma ; LeftUnital
\end{lstlisting}
\noindent is a useful intermediate definition.

The previous examples reinforce the importance of signatures, and of
morphisms from signatures to ``interesting'' theories as important, separate
entities.  For example, $\tmop{Monoid}$ as an \emph{embedding} is most
usefully seen as a morphism from $\tmop{PointedMagma}$.

Our machinery also allows one to construct the inverse view, from
\tmop{LeftUnital} to \tmop{RightUnital}.  Consider the view
\tmop{Flip ; LeftUnital} and the identity view from \tmop{LeftUnital}
to itself.  These are exactly the inputs for $\mathtt{mediate}$,
which returns a (unique) view from \tmop{LeftUnital} to \tmop{RightUnital}.
Furthermore, we obtain (from the construction of the mediating view) that this
view composes with that from \tmop{RightUnital} to \tmop{LeftUnital} to
give the identity.  This is illustrated in Figure~\ref{fig:RightUnital}
where the $\semB{-}$ annotations on nodes are omitted; note that the morphisms
are in $\catctx$, not $\catthry$. Let
\[ RU =
    \comfun\left(\semB{\tmop{Flip}}, \semE{\tmop{LeftUnital}}, \semP{id},
      \semP{\left[\mathtt{leftIdentity}\rename\mathtt{rightIdentity}\right]}
    \right) \]
then $\tmop{Flip}_{RU} = RU.\mathtt{view}_{\tmop{LeftUnital}}$
and
\[\tmop{Flip}_{LU} =
RU.\mathtt{mediate}_{\tmop{LeftUnital}}\left(
    \semE{\tmop{LeftUnital}} ; \semB{\tmop{Flip}}, \semE{id}
   \right)
\]
\noindent The constrution of $\mathtt{mediate}$ insures that
$\tmop{Flip}_{LU} ; \tmop{Flip}_{RU} = \semE{id}$, \emph{provided} that we
know that
\[\semB{\tmop{Flip}} ; \semB{\tmop{Flip}} = \semB{id} :
    \semC{\tmop{Magma}}\rightarrow\semC{\tmop{Magma}}. \]
\noindent The above identity is not, however, structural, it properly belongs
to the underlying type theory: it boils down to asking if
\[ \forall x:U. \mathtt{flip}\left(\mathtt{flip\ x}\right) =_{\beta\eta\delta} x \]
\noindent or, to use the notation of \S\ref{sec:pres},
\[ \left[ U:\mathtt{Type}, x:U\right] \vdash
\mathtt{flip}\left(\mathtt{flip\ x}\right) \equiv x : U.\]

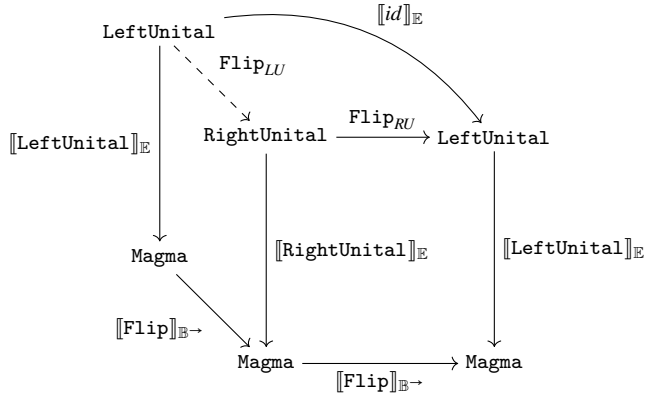
\begin{figure}[t]
\centering
\begin{tikzpicture}[node distance=3cm, auto]
  \node (P)              {$\tmop{RightUnital}$};
  \node (B) [right of=P] {$\tmop{LeftUnital}$};
  \node (A) [below of=P] {$\tmop{Magma}$};
  \node (C) [below of=B] {$\tmop{Magma}$};
  \node (X) [node distance=1.4cm, left of=P, above of=P] {$\tmop{LeftUnital}$};
  \node (Y) [below of=X] {$\tmop{Magma}$};
  \draw[->] (P) to node {$\tmop{Flip}_{RU}$} (B);
  \draw[->] (P) to node {$\semE{\tmop{RightUnital}}$} (A);
  \draw[->] (A) to node [swap] {$\semB{\tmop{Flip}}$} (C);
  \draw[->] (B) to node {$\semE{\tmop{LeftUnital}}$} (C);
  \draw[->] (X) to node [swap] {$\semE{\tmop{LeftUnital}}$} (Y);
  \draw[->, bend left] (X) to node {$\semE{id}$} (B);
  \draw[->, dashed] (X) to node {$\tmop{Flip}_{LU}$} (P);
  \draw[->] (Y) to node [swap] {$\semB{\tmop{Flip}}$} (A);
\end{tikzpicture}
\label{fig:RightUnital}
\caption{Construction of \tmop{LeftUnital} and \tmop{RightUnital}.  See
the text for the interpretation.}
\end{figure}

%\jacques{Do we really want to use these in this paper?  I understand how to
%do these now, but unlike the previous examples, I am not sure I could
%really explain what it is they show.}
%\begin{itemize}
%  \item Construction of diagonal view from carrier to double carrier as a
%  mediating morphism. \ This gives an example of building a mediating morphism.
%
%  \item Construction of magma from left operation and the diagonal view. This
%  gives an example of how to use the mediating morphism constructed above.
%\end{itemize}

\section{Discussion}\label{sec:discussion}

We have been careful to be essentially parametric in the
underlying type theory.  From a categorical point of view, this is hardly
surprising: this is the whole point of contextual categories~\cite{Cartmell}.
A numbers of features can be added to the type theory, at no harm to
the combinators themselves -- see Jacobs~\cite{JacobsCLTT} and
Taylor~\cite{taylor1999practical} for many such features.

One of the features that we initially built in was to allow \emph{definitions}
in our theory presentations.
This is especially useful when transporting theorems from
one setting to another, as is done when using the ``Little Theories''
method~\cite{LittleTheories}.  It is beneficial to first build up \emph{towers
of conservative extensions} above each of the theories, so as to build up a
more convenient common vocabulary, which then makes interpretations easier to
build (and use)~\cite{carette03calculemus}.
However, this complicated the meta-theory too much, and that feature has
been removed from the current version.  We hope to use ideas
from~\cite{TTFP} towards this goal.

Lastly, we have implemented a ``flattener'' for our semantics,
which just turns a presentation $A$ given in our language into a flat
presentation $\mathsf{Theory}\{\ l\}$ by computing
$\mathsf{cod}\left(\semE{A}\right)$.  We have been very careful to
ensure that all our constructions leave no trace of the construction
method in the resulting flattened theory.  We strongly believe that
users of theories do not wish to be burdened by such details, and we also
want developers to have maximal freedom in designing a modular, reusable
and maintainable hierarchy without worrying about backwards compatibility
of the \emph{hierarchy}, only the end results: the theory presentations.

\paragraph{Realization}
Three prototypes exist: a stand-alone version~\cite{CaretteOConnorTPC,TPCProto},
one as emacs macros aimed at Agda~\cite{next-700-git}, and another~\cite{diagrams_mmt}
in \mmt\cite{MMT}, 
a framework for developing logics which allows rapid
prototyping of formal systems~\cite{prototyping_mmt}.  
In \mmt, we define the combinators as new theory expressions, 
defined as untyped symbol declarations in a theory, see Figure~\ref{fig:combinators_theory}.
The syntax is not identical to that of the previous section, but it is quite
close.  These theory expressions are then
elaborated using rules implemented in Scala, the implementation language of
\mmt. A \verb|Combinators| theory specifies which rules are loaded when
it is being used. 

The rules implement the operational semantics presented in
Section~\ref{sec:combinators}.   As the results of the combinators are
records of information, we choose to elaborate the expressions into
theory diagrams (i.e. a directed graph of theories and theory morphisms), which
can be named~\cite{diagrams_mmt}. These diagrams have a \emph{distinguished theory}, corresponding
to our $\semC{-}$ semantics, and a \emph{distinguished morphism}, corresponding
to our $\semE{-}$ semantics. Using this implementation, the graph in our motivating example, 
Figure~\ref{fig:cube2_monoid}, is generated by the code in Figure~\ref{fig:monoid_mmt}. 

\begin{figure}
\begin{lstlisting}[mathescape]
theory Combinators =
  include ?ModExp 
  empty # EMPTY 1 
  extends # 1 EXTEND { %L1_L2, $\cdots$} prec -1000000
  rename1 # L1 $\to$ L2 prec -500000
  rename # 1 RENAME { 2,$\cdots$ } prec -1000000
  combine # COMBINE 1 {2,$\cdots$} 3 {4,$\cdots$} prec -2000000
  translate # MIXIN 1 {2,$\cdots$} 3 {4,$\cdots$} prec -2000000

  rule rules?ComputeEmpty 
  rule rules?ComputeExtends 
  rule rules?ComputeRename 
  rule rules?ComputeCombine 
  rule rules?ComputeMixin 
\end{lstlisting}
\caption{Defining the syntax for the theory expressions and the rules to elaborate them in MMT.}
\label{fig:combinators_theory}
\end{figure}
\begin{figure}
\begin{lstlisting}[mathescape]
theory Empty = 
diagram Carrier := ?Empty EXTEND { U : sort} 
diagram Pointed := ?Carrier_pres EXTEND { e : tm U } 
diagram Pointed0 := ?Pointed RENAME {e $\rightsquigarrow$ 0} 
diagram Magma := ?Carrier_pres EXTEND {$\circ$: tm U $\rightarrow$ tm U $\rightarrow$ tm U }
diagram MagmaPlus := ?Magma RENAME {$\circ$ $\rightsquigarrow$ +} 
diagram Semigroup := ?Magma_pres EXTEND { assoc : $\cdots$ } 
diagram SemigroupPlus := 
   COMBINE ?Semigroup {$\circ$ $\rightsquigarrow$ + , assoc $\rightsquigarrow$ assoc_+} ?MagmaPlus {} 
diagram PointedMagma := 
   COMBINE ?Magma {} ?Pointed {} 
diagram PointedMagmaPlus := 
   COMBINE ?PointedMagma {$\circ$ $\rightsquigarrow$ +, e $\rightsquigarrow$ 0} ?Pointed0 {} 
diagram RightUnital := ?PointedMagma_pres EXTEND { runital : $\cdots$ } 
diagram RightUnitalPlus := 
   COMBINE ?RightUnital {$\circ$ $\rightsquigarrow$ + , e $\rightsquigarrow$ 0 } 
?PointedMagmaPlus {} 
diagram LeftUnital := ?PointedMagma_pres EXTEND { lunital : $\cdots$ } 
diagram LeftUnitalPlus := 
   COMBINE ?LeftUnital {$\circ$ $\rightsquigarrow$ + , e $\rightsquigarrow$ 0 } ?PointedMagmaPlus {} 
diagram Unital := 
   COMBINE ?RightUnital {} ?LeftUnital {} 
diagram UnitalPlus := 
   COMBINE ?RightUnitalPlus {} LeftUnitalPlus {} 
diagram Monoid := 
   COMBINE ?Unital {} ?Semigroup {} 
\end{lstlisting}
\caption{Defining the syntax for the theory expressions and the rules to elaborate them in MMT.}
\label{fig:monoid_mmt}
\end{figure}	
% assoc : $\vdash$ forall U [x y z : tm U] $\doteq$ U ((x $\circ$ y)  $\circ$ z) (x $\circ$ (y $\circ$ z))
% runital : $\vdash$ forall U [x : tm U] $\doteq$ U (x $\circ$ e) 
% lunital : $\vdash$ forall U [x : tm U] $\doteq$ U (e $\circ$ x) x 

The work in~\cite{next-700-git} uses the combinators we present in this paper
to build a library of \emph{PackageFormer}s~\cite{alhassy2019}, an abstraction
over bundling mechanisms. The combinators are implemented in terms of
variationals that manipulate PackageFormers~\cite{meta-prim-blog}. The
implementation provided now is an emacs environment, as a step to have these
features as an Agda language extension.

\section{Related Work}\label{sec:related}

% Here, we have more about systems that provide combinators 
Building large formal libraries while leveraging the structure of the domain
is not new, and has been tackled by theoreticians and system builders alike.
We have been particularly inspired by the early work of Goguen and
Brustall on CLEAR~\cite{PuttingTheoriesTogether,SemanticsOfClear} and
OBJ~\cite{goguen1987Obj}. The semantics of their combinators is
given in the category of theories and theory morphisms.
ASL~\cite{wirsing1986structured} and Tecton~\cite{Tecton1,Tecton2} are
also systems that embraced the idea of theory expressions early on. Smith's
Specware~\cite{Smith93,Smith99} really focuses on using the structure
of theory and theory morphisms to compose specifications via applying refinement
combinators, and in the end to generate software that is
correct-by-construction.  These systems gave us basic operational ideas,
and some of the semantic tools we needed.  In particular, these systems use the
idea of ``same name, same thing'' to choose between the different potential
meanings of combining theories.  We implemented a first prototype%
~\cite{MathSchemeExper} which seemed promising while the library was small
(less than $100$ theories). As it grew, difficulties were encountered: basically
the ``same name, same thing'' methodology was not adequate for encoding
mathematics without a lot of unnecessary duplication of concepts.

A successor of ASL, CASL~\cite{CoFI:2004:CASL-RM}, and its current
implementation Hets~\cite{HETS} offers many more combinators than we do
for structuring specifications, as well as a documented categorical semantics.
To compare with our \tmop{combine}, CASL has a
\verb|sum| operation (on a ``same name, same thing'' basis) that
builds a colimit, similar to what is also done in Specware. However the 
use of ``same name, same thing'' is
problematic when combining theories developed independently that might 
(accidentally) use the
same name for different purposes. For example, extending \verb|Carrier| with two
different binary operations and their respective unit. If these units happened
to have been named the same (say~\verb|e|), then by using \verb|sum|, the resulting theory
would have a single ``point'' \verb|e|
that acts as the unit for both operations, which is not what either user
intended.  In other words, while the result is a pushout, it is not
the pushout over the intended base theory.  CASL structured
specifications have other complex features: they allow a user to specify
that a model of a specification should be \emph{free}, or to derive a
specification from a previous one via \emph{hiding} certain fields.
Semantically, these are very useful features to have.
Unfortunately, as the CASL manual also points out, the use of these features
makes it impossible to write a ``flattener'' for theories; in other words,
there is, in general, no means to have a finite theory presentation for these.
This goes against our usability requirement that users should be able to see a
fully flattened version of their theories.

% Here we talk about publications - not associated with a specific system -
% that share ideas with us 
Of the vast algebraic specification literature around this topic, we want to
single out the work of Oriat~\cite{Oriat} on isomorphism of specification
graphs as capturing similar ideas to ours on extreme modularity.  And it cannot
be emphasized enough how crucial Bart Jacob's book~\cite{JacobsCLTT} has been
to our work.

Another line of influence is through universal
algebra~\cite{whitehead1898treatise,burris1981course}, more precisely the
\emph{constructions} of universal algebra, rather than its theorems.  That we
can manipulate signatures as algebraic objects is firmly from that literature.
Of course, we must generalize from the single-sorted equational approach of the
mathematical literature, to the dependently typed setting, as
Cartmell~\cite{Cartmell} started and Taylor~\cite{taylor1999practical}
advocates.  As we eschew all matters dealing with models, the syntactic
manipulation aspects of universal algebra generalize quite readily.  The
syntactic concerns are also why \emph{Lawvere theories} ~\cite{LawvereTheories}
are not as important to us.
Sketches~\cite{barr1990category} certainly could have been used, but would have
led us too far away from the elegance of using structures already present in
the $\lambda$-calculus (namely contexts) quite directly.
The idea of defining  combinators over theories has also been used in Maude
\cite{duran2007maude} and Larch \cite{guttag2012larch}. 

\emph{Institutions} might also appear to be an ideal setting for our work.  But
even as the relation to categorical logic has been worked
out~\cite{GoguenMPRS07}, it remains that these theories are largely semantic,
in that they all work up to equivalence.  This makes the theory of institutions
significantly simpler, however it also makes it difficult to use for
user-oriented systems: people really do care what names their symbols have in
their theory presentations.

The Harper-Mitchell-Moggi work on Higher-order Modules%
~\cite{harper1989higher} covers some of the same themes we do: a
set of constructions (at the semantic level)
similar to ours is developed for ML-style modules.  However,
they did not seem to realize that these constructions could be turned into
an external syntax, with application to structuring a large
library of theories (or modules).  Nor did they see the use of fibrations,
since they avoided such issues ``by construction''.  Moggi returned to
this topic~\cite{Moggi:1989}, and did make use of fibrations as well
as \textit{categories with attributes}, a categorical version of contexts.
Post-facto, it is possible to recognize some of our ideas as being
present in Section 6 of that work; the emphasis is however completely
different. In that same vein \emph{Structured theory presentations and
logic representations}~\cite{HST1994} does have a set of combinators.
However, the semantics is inaccurate: Definition 3.3 of the signature of a
presentation requires that both parts of a union must have the same signature
(to be well formed) and yet their Example 3.6 on the next page is not
well formed. That being said, many parts of the theory-level semantics is
the same. However, it is our morphism-level semantics which really allows
one to build large hierarchies conveniently.

A rather overlooked 1997 Ph.D. thesis by Sherri Shulman~\cite{Shulman:1997}
presents a number of interesting combinators. Unfortunately the
semantics are unclear, especially in cases where theories have parts in
common; there are heavy restrictions on naming, and no renaming, which makes
the building of large hierarchies fragile.  Nevertheless, there is much
kinship here, especially that of extreme modularity. If this work had been
implemented in a mainstream tool, it would have saved us a lot of effort.

Taylor's magnificent ``Practical Foundations of
Mathematics''~\cite{taylor1999practical} does worry more about syntax.
Although the semantic component is there, there are no algorithms and no notion
of building up a library. The categorical tools are presented too, but
not in a way to make the connection sufficiently clear so as to lead to an
implementable design.  That work did lead us to investigate aspects of
Categorical Logic and Type Theory, as exposed by Jacobs~\cite{JacobsCLTT}.
This work and the vast literature on categorical approaches to
dependent type theory%
\cite{r2019typetheoretic,gambino2019models,JACOBS1993169,hofmann1997syntax,hofmann1995extensional,hofmann1994interpretation,dybjer1995internal,ahrens2017displayed,clairambault2014biequivalence,streicher2012semantics,firore2001semantics,fiore2005mathematical}
reveal that the needed structure really is already present, and just needs
to be reflected back into syntax.

MMT~\cite{MMT} shares the same motivation of building theories modularly with
morphisms connecting them. However, it does not support theory expressions.
Apart from inclusions, all morphisms need to be given manually ---
\cite{DHS2009} shows some examples.
Many of the scaling problems that we have identified
are still present. MMT does have some advantages: it is foundation 
independent, and possesses some rather nice web-based tools for pretty
display. But the extend operation (named \texttt{include}) is
theory-internal, and its semantics is not given through flattening. The result
is that their theory
hierarchies explicitly suffer from the ``bundling'' problem, as lucidly
explained in~\cite{SpittersvanderWeegen}, who introduce type classes in
Coq to help alleviate this problem. 

Coq has both \emph{Canonical Structures} and \emph{type classes}%
~\cite{SpittersvanderWeegen}, but no combinators to make new ones out of old.
Similarly, Lean~\cite{Lean} has some (still evolving) structuring
mechanisms but not combinators to form new theories from old.

Isabelle's \emph{locales} support \emph{locale expressions}%
\cite{ballarin2003locales}, which are also reminiscent of ours.  However, we
are unaware of a denotational semantics for them; furthermore, neither combine
nor mixin are supported; their merge operation uses a same-name-same-thing
semantics.  Axiom~\cite{jenks1992axiom} does support theory formation
operations, but these are quite restricted, as well as defined purely
operationally.  They were meant to mimic what mathematicians do informally when
operating on theories. To the best of our knowledge, no semantics for them has
ever been published.

\section{Conclusion}\label{sec:conclusion}

There has been a lot of work done in mathematics to give structure to
mathematical theories, first via universal algebra, then via category
theory.  But even though a lot of this work
started out being syntactic, very quickly it became mostly
semantic within a non-constructive meta-theory, and thus largely useless for
the purposes of concrete implementations and full automation.

Here we make the observation that, for dependent type theories in common use,
the category of theory presentations coincides with the opposite of the
category of contexts.  This allows us to draw freely from developments in
categorical logic, as well as to continue to be inspired by algebraic
specifications.  Interestingly, key here is to make the opposite choice as
Goguen's (as the main inspiration for the family of OBJ languages) in two ways:
our base language is firmly higher-order, as well as dependently typed, while
our ``module'' language is first-order, and we work in the opposite category.

We provide a simple-to-understand term language of ``theory expression
combinators'', along with their semantics.  We have shown
that these fit our requirements of allowing to capture mathematical
structure, while also allowing this structure to be hidden from users.

The design was firmly driven by its main application: to build a large
library of mathematical theories, while capturing the inherent structure
known to be present in such a development. To reflect mathematical practice,
it is crucial to \emph{take names seriously}. This leads us to ensuring
that renamings are not only allowed, but fundamental. Categorical
semantics and the desire to capture structure inexorably lead us towards
considering \emph{theory morphisms}
as the primary notion of study --- even though our original
goal was grounded in the theories themselves. We believe that other work on
theory combinators would have been more successful had they focused on the
morphisms rather than the theories; even current work in this domain,
which pays lip service to morphisms, is very theory-centric.

Paying close attention to the ``conventional wisdom'' of category logic
and categorical treatments of dependent type theory
led to taking both cartesian liftings and mediating morphisms as important
concepts. Doing so immediately improved the compositionality of our
combinators. Noticing that this puts the focus on fibrations was also
helpful. Unfortunately, taking names seriously means that the fibrations
are not cloven; we turn this into an opportunity for users to retain control
of their names, rather than to force some kind of ``naming policy''.

A careful reader will have noticed that our combinators are ``external'',
in the sense that they take and produce theories (or morphisms or ...).
Many current systems use ``internal'' combinators, such as \texttt{include},
potentially with post facto qualifiers (such as \texttt{ocaml}'s
\texttt{with} for signatures) to ``glue'' together items that would have
been identified in a setting where morphisms, rather than theories, are
primary. Furthermore, we are unaware of any system that guarantees that
their equivalent to our $\mathsf{combine}$ is \emph{commutative} 
(Proposition~\ref{combine-commutative}). Lastly, this enables future
features, such as limits of diagrams, rather then just binary
combinations/mixins.

Prototype implementations~\cite{TPCProto,diagrams_mmt,meta-prim-blog}
show the usefulness of our language; the first of these
was used to capture the knowledge
for most of the theories on Jipsen's list~\cite{jipsen} as well as
many others from Wikipedia, most of the modal logics on Halleck's
list~\cite{halleck}, as well as two formalizations of basic
category theory, once dependently-typed, and another following Lawvere's
ETCS approach as presented on the nLab~\cite{ETCS}. Totally slightly
over $1000$ theories in slightly over $2000$ lines of code, this
demonstrates that our combinators, coupled with the \emph{tiny theories}
approach, does seem to work.

Even more promising, our use of standard categorical constructions
points the way to simple generalizations which should allow us to capture
even more structure, without having to rewrite our library.  Furthermore,
as we are largely independent of the details of the type theory, this structure
seems very robust, and our combinators should thus port easily to other
systems.

\begin{acknowledgements}
We wish to thank Michael Kohlhase, Florian Rabe and William M. Farmer
for many fruitful conversations on this topic. The comments of one
referee were very useful in helping us clarify our paper, and we are
thankful to them.  We also wish to thank the
participants of the Dagstuhl Seminar 18341 ``Formalization of Mathematics
in Type Theory'' whose interest in this topic helped insure that this paper got
finished.
\end{acknowledgements}

\bibliography{mathscheme}
\bibliographystyle{spmpsci}

\end{document}